\documentclass[a4paper, 12pt]{article}

\usepackage[utf8]{inputenc}
\usepackage[T1]{fontenc}
\usepackage[a4paper, margin=2.5cm]{geometry}
\usepackage{needspace}

\usepackage{libertine} % text font
\usepackage{inconsolata} % texttt font

\usepackage{amsmath, amsthm, amssymb}
\usepackage{graphicx}
\usepackage{enumerate}
\usepackage{authblk}
\usepackage[textsize=footnotesize, textwidth=2cm, color=yellow]{todonotes}
\usepackage{thmtools}
\usepackage{thm-restate}
\usepackage[colorlinks=true, citecolor=red]{hyperref}
\usepackage{float}

\renewenvironment{abstract}
{\small\vspace{-1em}
\begin{center}
\bfseries\abstractname\vspace{-.5em}\vspace{0pt}
\end{center}
\list{}{
\setlength{\leftmargin}{0.6in}%
\setlength{\rightmargin}{\leftmargin}}%
\item\relax}
{\endlist}

\declaretheorem[name=Theorem, numberwithin=section]{theorem}
\declaretheorem[name=Lemma, sibling=theorem]{lemma}
\declaretheorem[name=Proposition, sibling=theorem]{proposition}
\declaretheorem[name=Definition, sibling=theorem]{definition}
\declaretheorem[name=Corollary, sibling=theorem]{corollary}
\declaretheorem[name=Conjecture, sibling=theorem]{conjecture}
\declaretheorem[name=Claim, sibling=theorem]{claim}

\def\cqedsymbol{\ifmmode$\lrcorner$\else{\unskip\nobreak\hfil
\penalty50\hskip1em\null\nobreak\hfil$\lrcorner$
\parfillskip=0pt\finalhyphendemerits=0\endgraf}\fi} 
\newcommand{\cqed}{\renewcommand{\qed}{\cqedsymbol}}

\def\AA{\mathbf{A}} % 
\def\BB{\mathbf{B}} % 
\def\CC{\mathbf{C}} % 
\def\D{\mathcal{D}} % dominating set
\def\G{\mathcal{G}} % hypergraph 1
\def\H{\mathcal{H}} % hypergraph 2
\def\E{\mathcal{E}} % hyperedge
\def\LL{\mathbf{L}}
\def\N{\mathcal{N}} % neighborhoods
\newcommand{\NN}{\mathbb{N}} % integers
\def\Oh{O} % big O notation
\def\UU{\mathbf{R}} % right
\def\TT{\mathbf{T}} % supertree
\def\UU{\mathbf{U}} % up
\newcommand{\algoa}{\texttt{A}} % algo A
\newcommand{\algob}{\texttt{B}} % algo B
\newcommand{\algoc}{\texttt{C}} % algo C

\newcommand{\intv}[2]{\left\{#1,\dots,#2\right\}}
\newcommand{\doa}[1]{{\downarrow\!#1}} % pretty arrow
\newcommand{\upa}[1]{{\uparrow\!#1}} % 

\DeclareMathOperator{\Min}{Min}
\DeclareMathOperator{\Max}{Max}
\DeclareMathOperator{\parent}{Parent}
\DeclareMathOperator{\child}{Children}
\DeclareMathOperator{\priv}{Priv}
\DeclareMathOperator{\poly}{poly}
\DeclareMathOperator{\MIS}{MIS}
\DeclareMathOperator{\Tr}{Tr}

\let\leq\leqslant
\let\geq\geqslant

\title{Enumerating minimal dominating sets\\ in the (in)comparability graphs\\ of bounded dimension posets\thanks{The first author has been supported by the ANR project GrR ANR-18-CE40-0032 (2019-2023). The second and last authors have been supported by the ANR project GraphEn ANR-15-CE40-0009. The third author has been supported by the Polish National Science Center grant (BEETHOVEN; UMO-2018/31/G/ST1/03718).}}

\author[1]{Marthe Bonamy}
\author[2]{Oscar Defrain}
\author[3]{\\Piotr Micek}
\author[4]{Lhouari Nourine}
\affil[1]{CNRS, LaBRI, Université de Bordeaux, France.}
\affil[2]{LIS, Aix-Marseille Université}
\affil[4]{LIMOS, Université Clermont Auvergne, France.}
\affil[3]{Theoretical Computer Science Department,\protect\\Faculty of Mathematics and Computer Science,\protect\\Jagiellonian University, Krak\'{o}w, Poland.}
\date{April 5, 2020}

\begin{document}

\maketitle

\begin{abstract}
Enumerating minimal transversals in a hypergraph is a notoriously hard problem. It can be reduced to enumerating minimal dominating sets in a graph, in fact even to enumerating minimal dominating sets in an incomparability graph.
We provide an output-polynomial time algorithm for incomparability graphs whose underlying posets have bounded dimension.
Through a different proof technique, we also provide an output-polynomial time algorithm for their complements, i.e., for comparability graphs of bounded dimension posets.
Our algorithm for incomparability graphs relies on the geometrical representation of incomparability graphs with bounded dimension, as given by Golumbic et al.\ in 1983. 
It runs with polynomial delay and only needs polynomial space.
Our algorithm for comparability graphs is based on the flipping method introduced by Golovach \mbox{et al.~in 2015}. 
It performs in incremental-polynomial time and possibly requires exponential space.
%
% In addition, we show how to improve the flipping method so that it requires only polynomial space. Since the flipping method is a key tool for the best known algorithms enumerating minimal dominating sets in a number of graph classes, this yields direct improvements on the state of the art. 

% plain text abstract below :

% Enumerating minimal transversals in a hypergraph is a notoriously hard problem. It can be reduced to enumerating minimal dominating sets in a graph, in fact even to enumerating minimal dominating sets in an incomparability graph. We provide an output-polynomial time algorithm for incomparability graphs whose underlying posets have bounded dimension. Through a different proof technique, we also provide an output-polynomial time algorithm for their complements, i.e., for comparability graphs of bounded dimension posets. Our algorithm for incomparability graphs relies on the geometrical representation of incomparability graphs with bounded dimension, as given by Golumbic et al. in 1983. It runs with polynomial delay and only needs polynomial space. Our algorithm for comparability graphs is based on the flipping method introduced by Golovach et al. in 2015. It performs in incremental-polynomial time and possibly requires exponential space.

\vskip5pt\noindent{}{\bf Keywords:} minimal dominating sets, poset dimension, comparability graphs, incomparability graphs, algorithmic enumeration
\end{abstract}

% edition
%\overfullrule=50pt % to spot overfull hbox

\section{Introduction}\label{sec:intro}

The problem we consider in this paper is the enumeration of all (inclusion-wise) minimal dominating sets of a graph, denoted by \textsc{Dom-Enum}.
A \emph{dominating set} in a graph $G$ is a set of vertices $D$ such that every vertex of $G$ is either in $D$, or is adjacent to a vertex in $D$.
It is (inclusion-wise) \emph{minimal} if no strict subset of $D$ is a dominating set.
Due to its equivalence with the problem of enumerating all (inclusion-wise) minimal transversals of a hypergraph, denoted by \textsc{Trans-Enum}, \textsc{Dom-Enum} has been widely studied this last decade. 
Latest publications on the problem include~\cite{golovach2018lmimwidth,kante2018holes,bonamy2019triangle,defrain2019neighborhood,golovach2019input}.

As an $n$-vertex graph $G$ may contain a number $|\D(G)|$ of minimal dominating sets which is exponential in $n$, one can only aim to provide algorithms running in time exponential in $n$, or, to provide algorithms whose running time is polynomial in $n+|\D(G)|$.
We only focus here on algorithms of the second type, and refer to~\cite{fomin2008combinatorial,couturier2013minimal,golovach2019input} for input exponential-time algorithms for the problem we consider in this paper.
An algorithm of the second type---running in polynomial time in the sizes of both the input and the output---is called \emph{output-polynomial}.
It is said to be running in \emph{incremental-polynomial time} if it moreover outputs the $i$-th solution in a time which is bounded by a polynomial in $n$ plus $i$, for all $i$.
If the running times before the first output, between any two consecutive outputs, and after the last output, are bounded by a polynomial in $n$ only, then the algorithm is said to be running with \emph{polynomial delay}.
We refer the reader to~\cite{johnson1988generating,creignou2019complexity,strozecki2019survey} for a more detailed introduction on the complexity of enumeration algorithms.

The existence of an output-polynomial time algorithm for \textsc{Dom-Enum} (or \textsc{Trans-Enum}) is a long-standing open question~\cite{eiter1995identifying,eiter2008computational,kante2014split}.
To date, the best known algorithm is due to Fredman and Khachiyan and comes from the dualization of monotone Boolean functions~\cite{fredman1996complexity}. 
It runs in output quasi-polynomial time $N^{o(\log N)}$ where $N=n+|\D(G)|$.
Output-polynomial time algorithms are known for several classes of graphs, including $\log(n)$-degenerate graphs~\cite{eiter2003new} and triangle-free graphs~\cite{bonamy2019triangle}, later extended to $K_t$-free, paw-free, and diamond-free graphs in \cite{bonamy2020kt}.
Incremental-poly\-nomial time algorithms are known for chordal bipartite graphs \cite{golovach2016chordalbip}, $\{C_6,C_8\}$-free bipartite graphs \cite{kante2018holes}, graphs of bounded conformality \cite{boros2004generating}, and unit square graphs \cite{golovach2018lmimwidth}.
Polynomial-delay algorithms are known for bounded degeneracy graphs~\cite{eiter2003new}, for chordal, and line graphs~\cite{kante2015chordal,kante2015line}.
Finally, linear-delay algorithms are known for graphs of bound\-ed cliquewidth \cite{courcelle2009linear}, permutation and interval graphs \cite{kante2013permutation}, split and $P_7$-free chordal graphs \cite{kante2014split,defrain2019neighborhood}. In terms of negative results, \textsc{Dom-Enum} is as hard in co-bipartite graphs as in general~\cite{kante2014split}. 

The graph classes we consider in the following are related to partially ordered sets, in short posets.
If $P$ is a poset on the set of elements $V$, then the \emph{comparability graph} of $P$ is the graph $G$ defined on vertex set $V(G)=V$ and where two distinct vertices $u$ and $v$ are adjacent if they are comparable in~$P$.
The \emph{incomparability graph} of $P$ is the complement: $u$ and $v$ are adjacent if they are incomparable in~$P$.
Note that every bipartite graph is a comparability graph, hence every co-bipartite graph is an incomparability graph.
The dimension, introduced in 1941 by Dushnik and Miller~\cite{dushnik1941partially}, is a key measure of complexity of posets and an analogue of the chromatic number for graphs.
The \emph{dimension} of a poset $P$ is the least integer $d$ such that elements of $P$ can be embedded into $\mathbb{R}^d$ in such a way that $x\leq y$ in $P$ if and only if the point of $x$ is below the point of $y$ with respect to the product order of $\mathbb{R}^d$ (the component-wise comparison of coordinates). 
% what about adding %seems reasonable 
Alternatively, the dimension of a poset can be defined as the least integer $d$ such that $P$ is the intersection of $d$ linear orders $\leq_1,\dots,\leq_d$ on same ground set, i.e., $x\leq y$ in $P$ if and only if $x\leq_1 y$, $\dots$, $x\leq_d y$.

Since \textsc{Dom-Enum} is as hard in incomparability graphs as in general, an output-polynomial algorithm would be a major, albeit unlikely, breakthrough.
Algorithms for natural subclasses of incomparability graphs, such as interval graphs (incomparability graphs of interval orders) and permutations graphs (incomparability graphs of $2$-dimensional posets), were obtained in~\cite{kante2013permutation}. Our contribution is as follows.

\needspace{2cm}

\begin{theorem}\label{thm:mainincomp}
    For any fixed integer $d$, there is a polynomial delay and space algorithm enumerating minimal dominating sets in the incomparability graphs of posets of dimension at most $d$, represented as the intersection of $d$ linear orders. % extensions
\end{theorem}

It should be noted here that the incomparability graphs of posets of dimension $d$ are $d$-trapezoid graphs, for which minimal dominating sets enumeration was already shown to be possible with linear delay after $n^{O(1)}$-time preprocessing, when given with their representation~\cite{golovach2018lmimwidth}.
We nevertheless point that our algorithm is remarkably simpler that the one obtained in \cite{golovach2018lmimwidth} which may be of practical interest for implementations purposes.
Our algorithm relies on the geometrical representation of incomparability graphs of bounded dimension, which was given by Golumbic et al.~in~\cite{golumbic1983comparability}.

\textsc{Dom-Enum} in comparability graphs is also widely open:
only the subcase of bipartite graphs (comparability graphs of posets of height at most $2$) has been solved recently~\cite{bonamy2019triangle}. 
However, there is no reason to believe that \textsc{Dom-Enum} is as hard for comparability graphs as in general. We prove that it is tractable when the underlying poset has bounded dimension, as follows.

\begin{theorem}\label{thm:maincomp}
    For any fixed integer $d$, there is an incremental-polynomial time algorithm enumerating minimal dominating sets in the comparability graphs of posets of dimension at most~$d$. 
\end{theorem}

The algorithm for comparability graphs is based on the \emph{flipping method} recently introducted by Golovach et al.~in~\cite{golovach2015flipping}.
It uses as a blackbox the algorithm of Khachiyan et al.~in \cite{khachiyan2007dualization} for the dualization in hypergraphs of bounded conformality and therefore may need space that is exponential in the size of the input.
It also covers the comparability graphs of $S_t$-free posets.

In addition to these contributions, we show with Lemma~\ref{lemma:il-flip-golovach-improved} how the flipping method can be improved to work only using polynomial space. While this is not enough to make Theorem~\ref{thm:maincomp} run in polynomial space, it turns existing algorithms for line graphs \cite{golovach2015flipping}, graphs of girth at least 7, chordal bipartite graphs \cite{golovach2016chordalbip}, and unit-square graphs \cite{golovach2018lmimwidth}, 
into incremental poly\-nomial time algorithms using polynomial space. 

The paper is organized as follows.
In Section~\ref{sec:preliminaries} we introduce definitions and notation from graph and order theory that will be used throughout the paper.
In Section~\ref{sec:polyflipping}, we recall the flipping method from~\cite{golovach2015flipping} and show how it can be improved to run using polynomial space.
We then show in Section~\ref{sec:flipping} how the method can be reduced to red-blue domination---a variant of domination---in comparability graphs.
In Section~\ref{sec:red-blue}, we show how red-blue domination can be solved in the comparability graphs of posets of bounded dimension, and conclude with our first algorithm.
The algorithm for incomparability graphs is presented in Section~\ref{sec:flashlight}.
We conclude the paper with discussions and open problems in Section~\ref{sec:conclusion}.

\section{Preliminaries}\label{sec:preliminaries}

The objects considered in this paper are finite.
If $A$ and $B$ are sets, we denote by $2^A$ the set of all subsets of $A$, and by $A\times B$ the Cartesian product $\{(a,b) \mid a\in A,\ b\in B\}$.

We start with some basic notions from graph theory.   
A \emph{graph} $G$ is a pair $(V(G),E(G))$ with $V(G)$ its set of vertices (or \emph{ground set}) and 
$E(G)\subseteq \{\{u,v\} \mid u,v\in V(G),\ {u\neq v}\}$ its set of edges.
Edges are denoted by $uv$ (or $vu$) instead of $\{u,v\}$.
Two vertices $u,v$ of $G$ are called \emph{adjacent} if $uv\in E(G)$.
A \emph{clique} (respectively an \emph{independent set}) in a graph $G$ is a set of pairwise adjacent (respectively non-adjacent) vertices.
A \emph{biclique} is a set of vertices that can be partitioned into two independent sets $A, B$ such that every vertex in $A$ is adjacent to every vertex in $B$.
We denote $K_t$ to be the clique on $t$ elements, and  $K_{t,t}$ to be the biclique on $2t$ elements of partition $A,B$ such that $|A|=|B|=t$.
The subgraph of $G$ \emph{induced} by $X\subseteq V(G)$, denoted by $G[X]$, is the graph $(X,E(G)\cap \{\{u,v\} \mid u,v\in X,\ u\neq v\})$; $G-X$ is the graph $G[V(G)\setminus X]$.
For every graph $H$, we say that $G$ is \emph{$H$-free} if no induced subgraph of $G$ is isomorphic to~$H$.

\paragraph{Domination}
Let $G$ be a graph and $u$ be a vertex of $G$.
The \emph{neighborhood} of $u$ is the set $N(u)=\{v\in V(G) \mid uv\in E(G)\}$.
The \emph{closed neighborhood} of $u$ is the set $N[u]= N(u)\cup\{u\}$.
For a subset $X\subseteq V(G)$ we define $N[X]=\bigcup_{x\in X} N[x]$ and $N(X)=N[X]\setminus X$.
Let $D,X\subseteq V(G)$ be two subsets of vertices of $G$.
We say that $D$ \emph{dominates} $X$ if $X\subseteq N[D]$.
It is (inclusion-wise) minimal if $X\not\subseteq N[D\setminus \{x\}]$ for any $x\in D$.
A (minimal) \emph{dominating set} of $G$ is a (minimal) dominating set of $V(G)$.
The set of all minimal dominating sets of $G$ is denoted by $\D(G)$. 
The problem of enumerating $\D(G)$ given $G$ by \textsc{Dom-Enum}.
For a graph $G$, we denote by $\MIS(G)$ the set of all its (inclusion-wise) maximal independent sets, and by \textsc{MIS-Enum} the problem of generating $\MIS(G)$ from $G$.
It is easily observed that every maximal independent set is a minimal dominating set, hence that $\MIS(G)\subseteq \D(G)$.
However, \textsc{MIS-Enum} appears to be much more tractable than \textsc{Dom-Enum}, as witnessed by the many efficient algorithms that are known for the problem~\cite{tsukiyama1977new,johnson1988generating,makino2004new}.
These last two observations are the starting point of the flipping method introduced in~\cite{golovach2015flipping} which we describe at Section~\ref{sec:polyflipping}.

Let $u\in D$. 
A vertex $v$ that is adjacent only to $u$ in $D$, i.e., for which $N[v]\cap D=\{u\}$, is a \emph{private neighbor} of $u$ with respect to $D$.
Note that $u$ can be its own private neighbor (we say that $u$ is \emph{self-private}).
The set of private neighbors of $u\in D$ is denoted by $\priv(D,u)$.
A set $I\subseteq V(G)$ is called \emph{irredundant} if $\priv(I,x)\neq\emptyset$ for all $x\in I$.
Then $D$ is a minimal dominating set of $G$ if and only if it is both a dominating set, and an irredundant set.

A graph $G$ together with two disjoint subsets $R,B\subseteq V(G)$ constitutes a \emph{red-blue graph} $G(R,B)$, where we consider vertices in $R$ to be \emph{red} and those in $B$ to be \emph{blue}. 
We do not require for $R$ and $B$ to partition $V(G)$. 
A \emph{red dominating set} of $G(R,B)$ is a red set $D\subseteq R$ that dominates $B$, i.e., that is such that $B\subseteq N[D]$.
It is (inclusion-wise) minimal if $B\not\subseteq N[D\setminus \{x\}]$ for any ${x\in D}$.
We denote $\D_G(R,B)$ to be the set of all minimal red dominating sets of $G(R,B)$, and \textsc{Red-Blue-Dom-Enum} the problem of enumerating $\D_G(R,B)$ given $G$, $R$, and~$B$. The index may be dropped when the graph $G$ is clear from the context. 
In the context of red-blue domination, we implicitly restrict our attention to the dominating sets that contain only red vertices, and to private neighbors that are blue.

\paragraph{Posets}

A \emph{partially ordered set} (or \emph{poset}) $P=(V,\leq)$ is a pair where $V$ is a set and $\leq$ is a  binary relation on $V$ that is reflexive, anti-symmetric and transitive.
Two elements $u$ and $v$ of $P$ are said to be \emph{comparable} if $u \leq v$ or $v \leq u$, otherwise they are said to be \emph{incomparable}, denoted $u \parallel v$.
A \emph{chain} in a poset $P$ is a set of pairwise comparable elements in $P$. %
An \emph{antichain} in a poset $P$ is a set of pairwise incomparable elements in $P$. 
A poset $P$ is called a \emph{total order} (or \emph{linear order}) if $V$ is a chain.
Posets are represented by their Hasse diagram; see Figure~\ref{fig:graph}.

The \emph{comparability graph} of a poset $P=(V,\leq)$ is the graph $G$ defined on same ground set $V(G)=V$ and where two vertices $u$ and $v$ are made adjacent if they are comparable in~$P$.
The \emph{incomparability graph} (or \emph{co-comparability graph}) of $P$ is the complement: $u$ and $v$ are made adjacent if they are incomparable in $P$, see Figure~\ref{fig:graph}.

We now define notations from order theory that will be used in the context of a poset and its \emph{comparability graph} only.
Let $G$ be the comparability graph of a poset $P=(V,\leq)$.
The \emph{ideal of $u$} is the set $\doa{u}=\{v\in V \mid v\leq u\}$, and the \emph{filter of $u$} is the dual $\upa{u}=\{v\in V \mid u\leq v\}$.
Note that $N[u]=\doa{u}\cup \upa{u}$.
These notions extend to subsets $S\subseteq V$ as follows: $\doa{S}=\bigcup_{u\in S} \doa{u}$, $\upa{S}=\bigcup_{u\in S} \upa{u}$. We note $\Min(S)$ and $\Max(S)$ the sets of minimal and maximal elements in $S$ with respect to $\leq$.
Clearly, $\Min(S)$ and $\Max(S)$ define antichains of $P$ for every $S\subseteq V$. 
In this context, a (minimal) dominating set of $G$ is a (minimal) set $D\subseteq V$ such that $\upa{D}\cup \doa{D}=V$, and a (maximal) independent set of $G$ is a (maximal) antichain of $P$.

\begin{figure}
    \centering
    \includegraphics[scale=1.2]{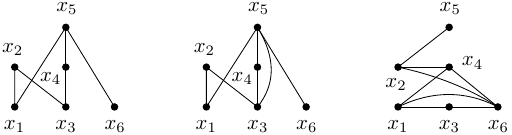}
    \caption{The Hasse diagram of a poset (left), its comparability graph (middle), and incomparability graph (right).}
    \label{fig:graph}
\end{figure}

The \emph{dimension} of a poset $P=(V,\leq)$ is the least integer $d$ such that $P$ is the intersection of $d$ linear orders $\leq_1,\dots,\leq_d$ on $V$, i.e., $x\leq y$ if and only if $x\leq_1 y$, $\dots$, $x\leq_d y$.

We call the poset induced by $X\subseteq V$, denoted $P[X]$, the suborder restricted on the elements of $X$ only.
A poset $P=(V,\leq)$ is \emph{bipartite} if $V$ can be partitioned into two sets $A,B$ such that $a<b$ implies $a\in A$ and $b\in B$.
We denote by $S_t$ and call \emph{standard example of order $t$} the poset with bipartition $A=\{a_1,\dots,a_t\}$ and $B=\{b_1,\dots,b_t\}$ such that $a_i<b_j$ for all $i \neq j\in\intv{1}{t}$.
See Figure~\ref{fig:posets} for an example.
It is well known that $S_t$ has dimension $t$~\cite{dushnik1941partially}.
Note that dimension is monotone under taking induced suborders. Therefore, posets containing a large standard example as a suborder have large dimension. The converse is however not true, as there are posets of unbounded dimension with no $S_2$ as an induced suborder, see~\cite{Tro-book} for a comprehensive study and references.

\begin{figure}[H]
    \centering
    \includegraphics[scale=1.2]{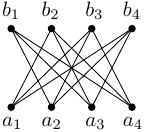}
    \caption{The standard example of order four.}
    \label{fig:posets}
\end{figure}

\paragraph{Hypergraphs}

A \emph{hypergraph} $\H$ is a pair $(V(\H),\allowbreak{}\E(\H))$ with $V(\H)$ its set of vertices (or \emph{ground set}) and 
$E(\H)\subseteq 2^{V(\H)}$ its set of edges (or hyperedges).
A hypergraph $\H$ is called \emph{Sperner} if $E_1\not\subseteq E_2$ for any two distinct hyperedges in $\H$.
A~\emph{transversal} of $\H$ is a set $T\subseteq V(\H)$ that intersects every hyperedge in~$\H$.
It is (inclusion-wise) minimal if $T\setminus \{x\}$ is not a transversal of $\H$ for any $x\in T$.
The sets of all minimal transversals of $\H$ is denoted by $\Tr(\H)$, and the problem of enumerating $\Tr(\H)$ given $\H$ by \textsc{Trans-Enum}.

Note that \textsc{Dom-Enum} appears as a particular case of \textsc{Trans-Enum} when considering the hypergraph $\N(G)$ of closed neighborhoods of a given graph $G$, defined by $V(\N(G))=V(G)$ and $\E(\N(G))=\{N[x] \mid x\in V(G)\}$.
Indeed, a minimal dominating set of $G$ is a minimal set intersecting every neighborhood of $G$, and hence $\Tr(\N(G))=\D(G)$.
In fact, the two problems were shown polynomially equivalent\footnote{\label{equivalence}%
Two enumeration problems $\Pi_A$ and $\Pi_B$ are said to be \emph{polynomially equivalent} when there is an output-polynomial time algorithm solving $\Pi_A$ if and only if there is one solving $\Pi_B$. 
If there is an output-polynomial time algorithm solving  $\Pi_A$ whenever there is one solving $\Pi_B$, then we say that $\Pi_B$ is \emph{(at least) as hard as} $\Pi_A$.}
by Kanté et al.~in \cite{kante2014split} even when restricted to co-bipartite graphs.
This in particular shows that \textsc{Dom-Enum} in co-bipartite graphs is as hard as for general graphs.
As for \textsc{Trans-Enum} and \textsc{Red-Blue-Dom-Enum}, it is easily seen that they are polynomially equivalent even when restricted to red-blue bipartite graphs of bipartition red and blue.
First, every red-blue graph $G(R,B)$ corresponds to a hypergraph $\H$, defined by  $V(\H)=R$ and $\E(\H)=\{N(x)\cap R \mid x\in B\}$, satisfying $\Tr(\H)=\D(R,B)$.
Then, every instance $\H$ of \textsc{Trans-Enum} corresponds to a bipartite instance $G(R,B)$ of \textsc{Red-Blue-Dom-Enum}, defined by $R=V(\H)$, $B=\{y_E \mid E\in \E(\H)\}$, with an edge $xy_E$ if and only if $x\in E$, which satisfies $\D(R,B)=\Tr(\H)$.

\section{Polynomial-space flipping method}\label{sec:polyflipping}

Golovach, Heggernes, Kratsch and Villanger introduced in~\cite{golovach2015flipping} the so-called \emph{flipping method} to efficiently enumerate minimal dominating sets in line graphs. 
This method was later used with much success~\cite{golovach2016chordalbip,golovach2018lmimwidth,kante2018holes}. 
A key step in this method is the \emph{flipping operation}. 
We recall it below and in the process, we show how the flipping method can be made to run with polynomial space (in contrast to exponential space from the original version).

\subsection{The flipping operation}\label{sec:flipping-operation}

Let~$G$ be an $n$-vertex graph and $v_1,\dots,v_n$ be any ordering of vertices in $G$.
We note that this order induces a lexicographical order on the family $2^{V(G)}$.
Let $D$ be a minimal dominating set of $G$ such that $G[D]$ contains at least one edge incident to some vertex $u$.
The following procedure is illustrated in Figure~\ref{fig:flipping-explained}.
Since $D$ is a minimal dominating set, the set $\priv(D,u)$ is not empty.
Let $v\in \priv(D,u)$. Since $u$ is adjacent with another vertex from $D$, we have $u\not\in \priv(D,u)$, so $v \neq u$. 
We want to replace $u$ with $v$ (\emph{flip} $u,v$) and obtain another minimal dominating set.
Let $X_{uv} \subseteq \priv(D,u)\setminus N[v]$ be the lexicographically smallest maximal independent set in $G[\priv(D,u)\setminus N[v]]$. 
Note that $X_{uv}$ is obtained from the empty set by iteratively adding a vertex of smallest index in $\priv(D,u)\setminus N[\{v\} \cup X_{uv}]$, until no such vertex exists.
Consider the set $D'=(D \setminus \{u\}) \cup X_{uv} \cup \{v\}$. 
Note that $D'$ is a (not necessarily minimal) dominating set of $G$. 
Some vertices of $D'$ may have no private neighbors, however every vertex of $X_{uv}\cup\{v\}$ is self-private.
Let $Z_{uv}$ be the lexicographically smallest set which has to be removed from $D'$ in order to make it minimal.
Similarly, $Z_{uv}$ is  obtained from the empty set by iteratively adding a vertex $z$ of smallest index in $D'\setminus Z_{uv}$ such that $z$ has no private neighbor with respect to $D' \setminus Z_{uv}$, until no such vertex exists. 
Since the elements of $X_{uv}\cup \{v\}$ are self-private in $D'$, the sets $X_{uv}\cup \{v\}$ and $Z_{uv}$ are disjoint, and non-adjacent.
Let us finally set $D^*=((D \setminus \{u\}) \cup X_{uv} \cup \{v\})\setminus Z_{uv}$.
Then $D^*$ is a minimal dominating set of $G$.

\begin{figure}
    \centering
    \includegraphics[scale=1.2]{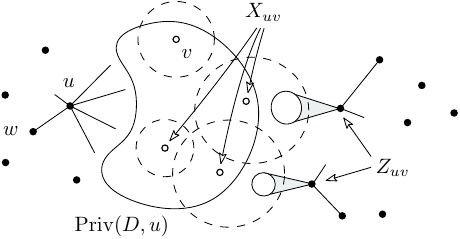}
    \caption{An illustration of the flipping operation on a dominating set $D$ such that $G[D]$ contains at least one edge incident to some vertex $u$, here depicted by $uw$. Black vertices are elements of $D$, white vertices are some elements of $\priv(D,u)$. Dashed discs represent closed neighborhoods, and plain disks represent private neighborhoods.}
    \label{fig:flipping-explained}
\end{figure}

Observe that since $X_{uv}$ and $Z_{uv}$ are selected greedily with respect to $v_1,\dots,v_n$, this procedure is deterministic.
Therefore, the procedure assigns to every minimal dominating set $D$ of $G$, and to every two vertices $u,v$ such that $u$ is in $D$ and is not isolated in $G[D]$, and $v\in \priv(D,u)$, a unique set~$D^*$.
We call $D^*$ the \emph{parent of $D$ with respect to flipping $u$ and $v$}, and denote it by $\parent_{uv}(D)$. 
The \emph{parent of $D$}, denoted by $\parent(D)$, is defined as the set $\parent_{uv}(D)$ for the smallest possible choice of a pair $u,v$; note that $D$ has a unique parent according to this definition.
Conversely, we denote by $\child(D^*)$ the set of all minimal dominating sets $D$ such that $D^*=\parent(D)$, and call \emph{child of $D^*$} any element of $\child(D^*)$. 
We point that the size of such a set may be exponential in $n$.
Note that, in the procedure, every edge in $G[D^*]$ is also an edge in $G[D]$, while there is at least one edge incident with $u$ that   appears in $G[D]$ and not in $G[D^*]$. 
This simple but important observation was formalized as follows.

\begin{proposition}[\cite{golovach2015flipping}]\label{prop:il-flip-golovach}
    Let $D,D^*\in \D(G)$ be such that $D^*=\parent_{uv}(D)$ for some edge $uv$.
    Then $E(G[D^*])\subsetneq E(G[D])$ and $v$ is an isolated vertex of $G[D^*]$.
\end{proposition}

The \emph{flipping operation} is then defined to be the reverse of how $D^*=\parent_{uv}(D)$ was generated from~$D$. 
This means, given $D^*$ with an isolated element $v\in D^*$, and $u$ a neighbor of $v$, the operation removes $X_{uv}\cup \{v\}$ and adds back $Z_{uv}\cup \{u\}$ to obtain a child~$D$ of $D^*$ with respect to flipping $u$ and~$v$. 
Obviously, the difficulty is to guess appropriate sets for $X_{uv}$ and $Z_{uv}$ when we are given only $D^*$, $u$, and $v$.

\subsection{The flipping method}\label{sec:flipping-method}

We now describe the flipping method as originally introduced in \cite{golovach2015flipping}.
Assume that there exists an algorithm $\algoa{}$ that, given $D^*\in \D(G)$, enumerates a family $\D$ of minimal dominating sets of $G$ such that $\child(D^*) \subseteq \D\subseteq \D(G)$. 
We stress the fact that $\D$ may contain minimal dominating sets that are not actual children of $D^*$.
The \emph{flipping method}, then, consists of a depth-first search (DFS) on a directed solution graph $\G$ whose nodes are minimal dominating sets of $G$, with one additional special node $r$, called the root, which has no in-neighbors.
The out-neighbors of the root are the maximal independent sets of $G$ (which are minimal dominating sets), and there is an arc from a minimal dominating set ${D^*\in V(\G)}$ to another one ${D\in V(\G)}$ if $\algoa{}$ generates $D$ from $D^*$.
At first, the DFS is initiated at the root.
Its out-neighbors are generated with polynomial delay using the algorithm of Tsukiyama et al.~\cite{tsukiyama1977new}.
The out-neighbors of the other nodes are generated using $\algoa{}$. 
Since $\algoa{}$ outputs (in particular) every child of a given node, we can argue using Proposition~\ref{prop:il-flip-golovach} that every minimal dominating set is reachable from $r$. 
More solutions may however be output by $\algoa{}$, and all the difficulty lies in handling the inherent repetitions.

In~\cite{golovach2015flipping} and later papers~\cite{golovach2016chordalbip,golovach2018lmimwidth}, a list of already visited nodes of $\G$ is maintained in order to handle repetitions, inexorably requiring space that is linear in $\D(G)$, thus potentially exponential in~$n$.
The achieved time complexity, on the other hand, is incremental-polynomial.

\begin{lemma}[\cite{golovach2015flipping}]\label{lemma:il-flip-golovach}
    Let $G$ be a graph.
    Suppose that there is an algorithm $\algoa{}$ that, given $D^*\in \D(G)$, enumerates with polynomial delay a family $\D$ of minimal dominating sets of $G$ such that $\child(D^*) \subseteq \D\subseteq \D(G)$.
    Then there is an algorithm that enumerates with incremental delay the set $\D(G)$ of all minimal dominating sets of $G$.
\end{lemma}

We would like to mention that a similar proof allows for ``incremental delay'' instead of ``polynomial delay'' in the hypothesis of this statement.
We further strengthen the statement in the following. 

\subsection{A polynomial-space flipping method}\label{sec:polyflipping-method}

We show here that guiding the DFS toward the children, together with a folklore trick (see e.g.~\cite{bonamy2020kt, capelli2023geometric} for recent formalizations) on running the algorithm again at each output, allows us to handle repetitions with polynomial space at the cost of an increased---but still incremental-poly\-nomial---complexity.

The next lemma is central to the next section and may be regarded as a space improvement of Lemma~\ref{lemma:il-flip-golovach}.
It is also of general interest as far as the flipping method is concerned.
% We note that in the statement below, we allow the delay $p$ between consecutive outputs to depend on the number of already output solutions; this is also known as \emph{incremental delay} in the literature~\cite{capelli2019incremental,capelli2023geometric} and yields incremental time with the number of solutions as an extra multiplicative factor. Moreover, if $p$ is a polynomial that does not depend on the number of solutions, then $p$ simply defines polynomial delay.
 
\begin{lemma}\label{lemma:il-flip-golovach-improved}
    Let $G$ be an $n$-vertex graph, and $p,s\colon \NN \to \NN$ be non-decreasing polynomial functions.
    Suppose that there is an algorithm $\algoa{}$ that, given $D^*\in \D(G)$, enumerates the $t$-th set of a family $\D$ satisfying $\child(D^*) \subseteq \D \subseteq \D(G)$ in time $p(n,t)$ and space $s(n)$.
    Then there is an algorithm that enumerates the $i$-th set of $\D(G)$ in
	\[
        \Oh(n^3)\cdot i^4 \cdot p(n,i)
    \] 
    time and $\Oh(n^2)\cdot s(n)$ space.
\end{lemma}

\begin{proof}
    In the following, let $\G'$ be the directed tree\footnote{This forms a subgraph of $\G$ as defined in Section~\ref{sec:flipping-method}: 
    Informally, the directed tree $\G'$ is what $\G$ would be if $\algoa{}$ was reliable, i.e., only generated children of $D^*$.} on vertex set $V(\G')=\D(G)\cup \{r\}$ and edge set $E(\G')=\{(r,D) \mid D\in \MIS(G)\}\cup \{(D^*,D) \mid D\in \child(D^*)\}$, where $r$ is a special vertex referred to as the \emph{root}. 
    	
    Let us first argue that every minimal dominating set $D$ is reachable from $r$ in $\G'$, by induction on the number of edges in it. 
    If $D$ contains no edge, it is a maximal independent set, thus an out-neighbor of $r$. 
    If $D$ contains an edge $uw$, consider such an edge with $u$ of smallest index in $D$, and its private neighbor $v$ of smallest index. 
    Let $D^*=\parent(D)$, i.e., $D^*$ is obtained by flipping $u,v$. 
    By Proposition~\ref{prop:il-flip-golovach}, $D^*$ has fewer edges than $D$ and is thus reachable from $r$.
    There is an arc from $D^*$ to $D$ in $\G'$, hence the conclusion.
    Therefore, a DFS of $\G'$ initiated at $r$ visits all minimal dominating sets of $G$. 
    
    Furthermore, for every minimal dominating set $D$, the length of the path from $r$ to $D$ is at most $|E(G[D])|\leq n^2$.

	We now describe an algorithm $\algob{}$ that enumerates, {possibly with repetitions}, the set $\D(G)$ of all minimal dominating sets of $G$.
	When $\algob{}$ outputs a set that was not output before we call this output a \emph{first occurrence}.
	The algorithm $\algob{}$ will output the $i$-th first occurrence (and so distinct solution) in
    \[
        \Oh(n^3)\cdot i^2\cdot p(n,i)
    \] 
    time and $\Oh(n^2) \cdot s(n)$ space.
    Algorithm $\algob{}$ proceeds as follows.
    First, it outputs every out-neighbor of $r$ without duplication using the algorithm of Tsukiyama et al.~in \cite{tsukiyama1977new}.
    Then, it proceeds with what boils down to a DFS (with additional outputs along the way) of $\G'$ initiated at $r$, as follows.  
    When visiting a node $D^*\in V(\G')$, $\algob{}$ seeks the children of $D^*$ by running $\algoa$.
    Each set $D$ returned by $\algoa$ is first output by $\algob{}$. 
    Then $\algob{}$ checks whether $D$ is a child of $D^*$. 
    If so, $\algob{}$ ``pauses'' the execution of $\algoa$ on $D^*$, and launches $\algoa$ on $D$. 
    When the execution of $\algoa$ on $D$ is complete, $\algob$ ``resumes'' the execution of $\algoa$ on $D^*$.

    Before, we discuss the time to compute first occurrences of $\algob{}$, 
    we take a pause to determine the following: 
    given $D$ and $D^*$ in $\G'$, how fast can we determine if $D$ is a child of $D^*$?
    The naive approach we choose goes as follows: 
    (1) consider the vertex of smallest index $u$ in $D$ such that there is an edge incident to $u$ in $G[D]$; 
    (2) select the vertex $v$ of smallest index in $\priv(D,u)$; 
    (3) perform the flipping operation along $u,v$;
    (4) check if the resulting set is $D^*$.
    Note that the selection of (1) and (2) can be done in $O(n^2)$ time.
    In order to make a flip when $u$ and $v$ are fixed, we need to compute the sets $X_{uv}$ and $Z_{uv}$. 
    The straightforward approach does it in $\Oh(n^3)$ time.
    Therefore, we can determine if $D$ is a child of $D^*$ in $\Oh(n^3)$ time and $\Oh(n^2)$ space 
    (as for convenience, we work with the adjacency matrix).
    
	We now examine the time spent by $\algob$ to produce its $i$-th first occurrence.
    The outputs generated by the algorithm of Tsukiyama et al., so the maximal independent sets of $G$, are all first occurrences and
	are produced with $\Oh(n^3)$ delay and $\Oh(n^2)$ space, hence within the claimed time bound; see~\cite{tsukiyama1977new}.

    Let $D$ be a first occurrence output by $\algob{}$ that is produced by a call of $\algoa$ on a node $D^*$ in $\G'$. 
    Say that $D$ is the ${(i-1)}$-th first occurrence in order output by $\algob{}$.
    Note that $D$ may or may not be a child of $D^*$.
    To obtain the next first occurrence output by $\algob{}$, 
    we consider the path from $r$ to $D$ in $\G'$ in case $D$ is a child of $D^*$, and the path from $r$ to $D^*$ in $\G'$ otherwise.
    The algorithm $\algob{}$ runs (or continue running) $\algoa{}$ on the last node of this path ($D$ or $D^*$), and
    for each node of the path, except $r$,
    $\algob{}$ has on the stack a paused execution of $\algoa{}$ called on the node.
    In the worst case scenario, all these executions will be resumed and each of them will spend at most $p(n,i-1)$-time, hence $p(n,i)$-time as $p$ is increasing, producing solutions that were already output by $\algob{}$ before.
    Moreover, if $D$ is a child of $D^*$, the execution of $\algoa{}$ on $D$ may produce a child $D'$ of $D$ that was already output by $\algob{}$ before.
    In that later case, the execution of $\algoa{}$ is in turn paused, a recursive call is made on $D'$, and so on.
    However, this can happen only $i$ times since the nodes of $\G'$ we explore that way are all distinct.
    So in total, at most $i$ such executions of $\algoa{}$ may, in a time bounded by $p(n,i)$ each, produce at most $i$ solutions that were already obtained by $\algob{}$. 
    Hence the time spent by $\algob{}$ between the $(i-1)$-th and the $i$-th first occurrences is bounded by $i\cdot p(n,i)$. 
    Each obtained set is checked to see if it is a child of the current node of $\G'$ in $\Oh(n^3)$ time. 
    Summing from $1$ to $i$, $\algob{}$ produces its $i$-th first occurrence or conclude that the $(i-1)$-th first occurrence was the last solution in a time which is bounded by $O(n^3)\cdot i^2\cdot p(n,i)$.
    Moreover, within this time bound, the total number of sets output by $B$ (with repetitions) is at most $i^2$.
	
	We note that since $p$ is a polynomial and $i\leq |\D(G)|\leq 2^n$, there is a small constant $c$ such that $\algob{}$ runs for at most $\poly(n+2^n) \leq 2^{cn}$ time.
	
	The space consumed by $\algob{}$ is dominated by the space taken by at most $n^2$ paused executions of $\algoa{}$ and the adjacency matrix of $G$. Thus $\algob{}$ runs in $\Oh(n^2) \cdot s(n)$ space.
	
	We are now ready to describe an algorithm $\algoc{}$ that enumerates $\D(G)$ without repetitions and within the desired time and space constraints.
	Algorithm $\algoc{}$ proceeds as follows.
	First, it launches a master instance of $\algob{}$. 
	It also maintains a counter keeping track of the number of steps (i.e., elementary steps counted by the time complexity) of the master instance of $\algob$. 
	Since $\algob{}$ runs for at most $2^{cn}$ steps, 
	% the counter of $cn$-bits long suffices.
	a $cn$-bits long counter suffices.
	Whenever the master instance outputs a new set $D$, and the counter of steps indicates $j$, 
	the algorithm $\algoc{}$ launches a new instance of $\algob{}$, runs it for $j-1$ steps, and compares each of its output with $D$.
	The new instance of $\algob{}$ is killed after exactly $j-1$ steps.
	If $D$ did not appear as the output of the new instance, 
	then we conclude that it is a first occurrence of the master instance, and 
	the algorithm $\algoc{}$ outputs $D$. 
	If $D$ has appeared as one of the outputs of the new instance, then $\algoc{}$ ignores it and continues the simulation of the master instance.
	In that way, every set of $\D(G)$ is output by $\algoc$ without repetitions.

	We now examine the time spent by $\algoc$ to produce the $i$-th minimal dominating set.
	Recall that at most $i^2$ sets are output by the execution of $\algob$ up to the $i$-th first occurrence.    
	Thus, the number of new instances of $\algob$ launched by $\algoc$ is bounded by $i^2$.
	Every such instance runs for at most $\Oh(n^3) \cdot i^2\cdot p(n,i)$ total time, as $p$ is non-decreasing.
	Hence in total, $\algoc$ runs for at most
	\[
	    \Oh(n^3) \cdot i^2\cdot p(n,i) + 
	    i^2 \cdot \Oh(n^3) \cdot i^2\cdot p(n,i) = \Oh(n^3) \cdot i^4 \cdot p(n,i)
	\]
	time steps to produce the $i$-th solution.
	
	The space consumed by $\algoc{}$ is determined by the space required by at most two independent instances of $\algob$ running in the same time which is $2 \cdot \Oh(n^2) \cdot s(n)$, and the size of the counter which is $O(n)$. 
	Thus $\algoc$ runs in $\Oh(n^2) \cdot s(n)$ space.
	\end{proof}

The algorithms given in~\cite{golovach2015flipping,golovach2016chordalbip,golovach2018lmimwidth} for line graphs, graphs of girth at least 7, chordal bipartite graphs, and unit-square graphs rely on the flipping method and run in incremental-polynomial time and exponential space. By directly plugging in Lemma~\ref{lemma:il-flip-golovach-improved} instead of Lemma~\ref{lemma:il-flip-golovach} in the procedure, we obtain the following.

\begin{corollary}\label{cor:flipamor}
    There is an incremental-polynomial time and polynomial-space algorithm enumerating minimal dominating sets in line graphs, graphs of girth at least 7, chordal bipartite graphs, and unit-square graphs.
\end{corollary}

\section{Flipping method in comparability graphs}\label{sec:flipping}

We now show how the flipping method, and more particularly the existence of an algorithm as required in Lemma~\ref{lemma:il-flip-golovach-improved}, can be reduced to red-blue domination in comparability graphs.

Recall that Lemma~\ref{lemma:il-flip-golovach-improved} is stated for general graphs and that the family $\D$ to be constructed can contain arbitrarily many solutions that are not actual children.
In~\cite{golovach2015flipping},~\cite{golovach2016chordalbip} and~\cite{golovach2018lmimwidth}, the authors were able to provide such an algorithm $\algoa{}$ in line graphs, graphs of girth seven, chordal bipartite graphs and unit square graphs. 
In these last two cases, they proved that to obtain an efficient $\algoa{}$, it suffices to design an efficient algorithm that enumerates all the minimal red dominating sets of an appropriate subgraph within the same class.
We conduct a similar analysis to show that, in comparability graphs, it suffices to design an efficient algorithm that enumerates all the minimal red dominating sets of a subgraph in which blue vertices are minimal with respect to the associated poset. 

Intuitively, the reason is that when flipping two adjacent vertices $u,v$ with, say, $u\leq v$, the set of private neighbors $Y$ that the set $X_{uv}\cup \{v\}$ may steal are by definition incomparable to $u$, hence smaller than $v$.
Moreover, the vertices in $Z_{uv}$ that had their private neighbors within $Y$ are non-adjacent to $v$, hence larger that $Y$ in the poset.
Consequently, in the reverse operation, and in order to guess the appropriate sets $X_{uv}$ and $Z_{uv}$, we may focus on minimal sets that dominate $Y$ from ``above'' and hence may restrict ourselves to the maximal elements in $Y$.
This yields the desired red-blue domination instance.

\begin{figure}
    \centering
    \includegraphics[scale=1.3]{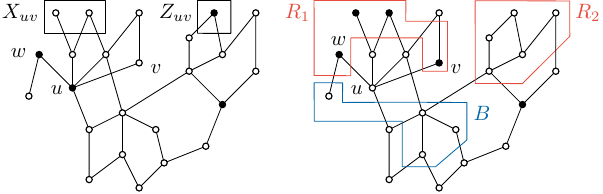}
    \caption{A minimal dominating set $D$ (on the left) and its parent $D^*=\parent_{uv}(D)$ (on the right) represented by black vertices in the Hasse diagram of the underlying poset of a comparability graph.
    The edge induced by $D$ and incident to $u$ is depicted by $uw$.}
    \label{fig:flipping}
\end{figure}

\needspace{4cm}

\begin{lemma}\label{lemma:flipping-to-redblue}
    Let $G$ be the comparability graph of a poset $P=(V,\leq)$.
    Suppose that there is an algorithm $\algob{}$ that, given an antichain $B$ of $P$ and a set ${R\subseteq\upa{B}\setminus B}$, enumerates with polynomial delay and polynomial space the set $\D(R,B)$ of minimal red dominating sets of $G(R, B)$.
    Then there is an algorithm $\algoa'$ that, given $D^*\in \D(G)$ and $u,v\in V(G)$, enumerates with polynomial delay and polynomial space a family $\D\subseteq \D(G)$ containing all $D$ such that $D^*=\parent_{uv}(D)$.
\end{lemma}

\begin{proof}
    The proof is conducted in the fashion of~\cite{golovach2016chordalbip}.
    We are given a minimal dominating set $D^*$ of $G$, an isolated vertex $v$ of $G[D^*]$, and a neighbor $u$ of $v$.
    Let us assume in the following that $u\leq v$ in $P$.
    The dual situation is handled by flipping upside-down the poset.
    This situation is depicted in Figure~\ref{fig:flipping} (right).
    We aim to compute using $\algob$ a family of sets $\D$ such that $\{D\in \D(G) \mid D^*=\parent_{uv}(D)\}\subseteq \D\subseteq \D(G)$.
    Recall that the difficulty of computing such sets given $D^*$, $u$, and $v$, lies in guessing appropriate sets for $X_{uv}$ and $Z_{uv}$ in the definition of the $\parent$ relation; see Section~\ref{sec:flipping-operation}.
    To this purpose, we now define sets that we will then show to be containing $X_{uv}$ and~$Z_{uv}$.
    
    Let $R_1\subseteq \upa{u}\cap D^*$ be the set of upper-neighbors $x$ of $u$ in $D^*$ such that $x\in \priv(D^*,x)$, i.e., these vertices are self-private in $D^*$. 
    Note that in particular, $R_1$ is an antichain of $P$, and it contains $v$.
    Let $B= V(G) \setminus N[(D^*\setminus R_1) \cup \{u\}]$ be the set of all vertices that are not dominated anymore when replacing $R_1$ with $u$ in $D^*$. 
    Note that $B \subseteq \doa{R_1}\setminus R_1$ since $R_1\subseteq \upa{u}$.
    Finally, let $R_2=\upa{B}\setminus N[R_1]$. 
    In particular, there are no edges between $R_1$ and $R_2$.
    Let us finally set $R=(R_1 \setminus \{v\}) \cup R_2$, and note that $B\subseteq \doa{R}\setminus R$.
    Informally, $R$ forms the set of vertices we can use to dominate $B$.
    We exclude $v$ from that set, since the whole point of the operation is to delete $v$. 
    Note that $B$ and $R$ may be empty; we will focus on these cases later.
    Also $B$ is not necessarily an antichain.
    We now show that we can restrict our attention to the maximal elements of $B$, and to the part of $R$ that is adjacent to these elements.

    The notation $\D(R\cap \upa{\Max(B)}, \Max(B))$ below is blatantly cumbersome. 
    To simplify the upcoming arguments, we prove that $\D(R\cap \upa{\Max(B)}, \Max(B))$ is in fact equal to the conceptually simpler $\D(R,B)$.

    \begin{claim}\label{claim:red-blue}
        The sets $\D(R, B)$ and $\D(R\cap \upa{\Max(B)}, \Max(B))$ are equal.
    \end{claim}
    
    \begin{proof}
        Recall that $B\subseteq \doa{R} \setminus R$.
        We first note that every dominating set of $\Max(B)$ in $R$ is a dominating set of $B$. 
        Indeed, for every $x \in B$ there exists $y \in \Max(B)$ such that $x\leq y$, and for any $x \in B$, $y \in \Max(B)$, and $z \in R$, if $x\leq y$ and $yz$ is an edge, then $y\leq z$ and $xz$ is an edge. 
        This guarantees $\D(R, \Max(B)) \subseteq \D(R, B)$ as if an element of $R$ has a private neighbor in $\Max(B)$, then it has in particular a private neighbor in $B$.

        The converse inclusion is straightforward in the sense that every dominating set of $B$ is in particular a dominating set of $\Max(B)$. As we argued above, any subset of $R$ that dominates $\Max(B)$ dominates $B$. Therefore, by minimality, no proper subset of a set in $\D(R,B)$ dominates $\Max(B)$. Hence $\D(R, B) \subseteq \D(R, \Max(B))$ and so $\D(R, B)=\D(R, \Max(B))$.
        
        We conclude by observing that the elements in $R\setminus \upa{\Max(B)}$ are not adjacent to the elements in $\Max(B)$, hence that they do not appear in any minimal dominating set of $G(R,B)$.
        Consequently $\D(R, B)=\D(R\cap \upa{\Max(B)}, \Max(B))$.
        \cqed
    \end{proof}

    Let us now describe $\algoa'$. 
    First if $R=\emptyset$ and $B\neq \emptyset$, then $\D(R,B)=\emptyset$ and we return $\D=\emptyset$ as no $D$ satisfies $D^*=\parent_{uv}(D)$.
    Indeed in that case, $R_1=\{v\}$, $B$ consists of the vertices that are not dominated anymore when replacing $v$ with $u$ in $D^*$, and since $B\subseteq \doa{v}$ and $R_2=\upa{B}\setminus N[v]=\emptyset$, every attempt to dominate $B$ by adding vertices to $D'=(D^*\setminus \{v\})\cup \{u\}$ will remove $v$ from the private neighbors of $u$ in $D'$.
    Since by definition $v$ is selected as a private neighbor of $u$ during the flipping operation, flipping $u$ and $v$ cannot produce a child of $D^*$ in that case. 
    Otherwise $\D(R,B)\neq \emptyset$, and we enumerate all minimal red dominating sets in $\D(R,B)$ using $\algob$ with Claim~\ref{claim:red-blue}.
    For each minimal red dominating set $X\in \D(R, B)$, we consider the set $D'=(D^*\setminus R_1)\cup \{u\} \cup X$ of vertices of $G$. 
    Note that $X$ may be empty, in which case $B= \emptyset$ and $\D(R, B)=\{ \emptyset \}$; that is not an issue. 
    We greedily reduce $D'$ into an irredundant set $D$ of $G$, and output $D$.
    
    This may seem counter-intuitive in an enumeration context, as a greedy reduction typically does not explore all options. 
    However, we will argue later (see~Claim~\ref{claim:child}) that $D'$ is already irredundant in all relevant cases, so $D=D'$ and the greedy reduction does not affect the pool of children.

    Let $\D$ be the multiset of all generated sets; we prove in the following four claims that $\D$ is in fact a set (no repetitions are produced) and that it has the desired properties. Namely, the correctness of $\algoa'$ follows from Claims~\ref{claim:distinct} and~\ref{claim:child}. We conclude the proof with the complexity analysis of $\algoa'$.
    
    \begin{claim}\label{claim:distinct}
    All outputs produced by the above procedure are distinct. Moreover the elements of the set $\D$ are minimal dominating sets of $G$ and $|\D|=|\D(R,B)|$.
    \end{claim}
    
    \begin{proof}
        As there is a natural bijection between minimal red dominating sets $\D(R,B)$ and distinct outputs of $\D$, we only need to argue two things: that every output is a minimal dominating set, and that there is no repetitions.
        
        There is nothing to argue in the case where $\D(R,B)=\emptyset$, and we assume from now on that $\D(R,B)$ is non-empty.    Let $X \in \D(R,B)$. To argue that its corresponding output is a minimal dominating set, it suffices to argue that $(D^*\setminus R_1)\cup \{u\} \cup X$ is a dominating set since we output an irredundant subset of it. Let $w$ be a vertex not dominated by $(D^*\setminus R_1)\cup \{u\}$. By definition, it belongs to $B$, so $w \in N[X]$. Therefore, $(D^*\setminus R_1)\cup \{u\} \cup X$ is a dominating set, as desired.
        
        Finally, observe that when greedily reducing $(D^*\setminus R_1)\cup \{u\} \cup X$ into a minimal dominating set $D$, we maintain $X \subseteq D$ as each element of $X$ has a private neighbor in $B$. In fact, we have $D \cap R = X$, which guarantees that a different choice of $X$ would yield a different output~$D$.
        \cqed
    \end{proof}

    The next claim shows that the sets $X_{uv}$ and $Z_{uv}$ to guess lie in $R_1$ and $R_2$, respectively; see Figure~\ref{fig:flipping-explained} for an example.

    \begin{claim}\label{claim:child1}
        For any set $D\in \D(G)$ such that $D^*=\parent_{uv}(D)$, let $X_{uv}$ and $Z_{uv}$ be the disjoint sets defined in the $\parent_{uv}$ relation, so that $D = (D^* \cup  \{u\} \cup Z_{uv}) \setminus (X_{uv} \cup \{v\})$. We have $X_{uv} \subseteq R_1 \setminus \{v\}$ and $Z_{uv}\subseteq R_2$. Additionally, for $\smash{Y_{uv}=\bigcup_{z\in Z_{uv}}\priv(D,z)}$, we have $Y_{uv} \subseteq B$.
    \end{claim}
    
    \begin{proof}
        Recall that $Z_{uv}$ is defined as a set of vertices that lose their private neighbors with respect to $D$ when adding $X_{uv}\cup \{v\}$ to $D\setminus \{u\}$.
        These private neighbors are the elements of the set $Y_{uv}$.
        
        By definition of the $\parent$ relation, $G[D]$ contains an edge $uw$, and $v$ is selected in the set $\priv(D,u)$.
        Since $uw$ is an edge, one of $u\leq w$ and $w \leq u$ holds. As $v \in \priv(D,u)$ and $w \in D$, the vertices $v$ and $w$ are incomparable. Since $u \leq v$, the case $w \leq u$ would lead to a contradiction, and we derive $u \leq w$.
        
        Let us first argue that $X_{uv} \subseteq R_1 \setminus \{v\}$. We have $v \not\in X_{uv}$, so we focus on proving $X_{uv} \subseteq R_1$. Recall that $X_{uv}\subseteq \priv(D,u)\setminus N[v]$ by definition. Since $u \leq v$, $X_{uv}$ cannot be in $\downarrow u$ and so we derive $X_{uv} \subseteq \upa{u} \cap D^*$. It remains to argue that $x \in \priv(D^*,x)$ for every $x \in X_{uv}$. Since $X_{uv}$ is an independent set by construction, we have $x \in \priv(X_{uv},x)$. Since $X_{uv} \cap N[v] = \emptyset$, we have  $x \in \priv(X_{uv} \cup \{v\},x)$. Since $X_{uv} \subseteq \priv(D,u)$, we derive $x \in \priv(X_{uv} \cup \{v\} \cup (D \setminus \{u\}),x)$, and so $x \in \priv(D^*,x)$ as $D^*\subseteq X_{uv} \cup \{v\} \cup (D \setminus \{u\})$.
        Hence $x \in R_1$ and it follows that $X_{uv}\subseteq R_1$.
        
        Let us now argue that $Z_{uv} \subseteq R_2$ and $Y_{uv} \subseteq B$. 
        Consider $z\in Z_{uv}$, and a private neighbor $y$ of $z$ with respect to $D$. Note that $y \in Y_{uv}$ and that $y$ is considered without loss of generality since every element of $Y_{uv}$ is the private neighbor of some element in $Z_{uv}$ with respect to $D$. Therefore, it suffices to argue that $z \in R_2$ and $y \in B$. 
        Recall that $z$ has no private neighbor with respect to $(D\setminus \{u\})\cup X_{uv}\cup \{v\}$, though $y$ is a private neighbor of $z$ with respect to $D$, which contains $u$.
        Since $X_{uv}\cup\{v\}$ and $Z_{uv}$ are non-adjacent, it follows that $z$ is not self-private (so $y\neq z$), and that $y$ is in the neighborhood of $X_{uv}\cup \{v\}$, but not in that of $u$.
        Since $X_{uv} \cup \{v\} \subseteq \upa{u}$ and $y \not \in N(u)$, we have $y\not\in \upa{(X_{uv}\cup \{v\})}$.
        We derive that $y\in \doa{(X_{uv}\cup \{v\})}\setminus N[u]$.
        If $y\geq z$, then $z\in \doa{(X_{uv}\cup \{v\})}$, which contradicts the fact that $X_{uv}\cup\{v\}$ and $Z_{uv}$ are non-adjacent.
        So $y\leq z$.
        Observe now that $z\not\in N[R_1]$.
        Otherwise as $z\not\in N[X_{uv}\cup\{v\}]$, either $z\in \doa{R_1\setminus (X_{uv}\cup\{v\})}$ and then $y\in \doa{R_1\setminus (X_{uv}\cup\{v\})}$, or $z\in \upa{R_1\setminus (X_{uv}\cup\{v\})}$ and there is an element $r$ in $R_1\setminus (X_{uv}\cup\{v\})$ such that $u\leq r\leq z$, hence such that $\priv(\{u,r,z\},r)=\emptyset$.
        These two situations lead to a contradiction since the elements of $R_1\setminus (X_{uv}\cup\{v\})$ are in $D$.
        So $z\not\in N[R_1]$ and since $R_2=\upa{B} \setminus N[R_1]$, the fact that $z \in R_2$ follows from $y \in B$, which we argue now.
        
        We have $N(y) \cap D^* \subseteq X_{uv}\cup\{v\}$, as the only neighbor of $y$ in $D$ is $z$. As shown earlier, $X_{uv} \subseteq R_1$, hence $N(y)\cap D^*\subseteq R_1$.
        Since $u \not\in N(y)$ and $B=V(G) \setminus N[(D^* \setminus R_1) \cup \{u\}]$, we derive $y\in B$, as desired. It follows that $Y_{uv} \subseteq B$ and  $Z_{uv}\subseteq R_2$.
        \cqed
    \end{proof}

    The next claim shows that the set $X_{uv}$ to remove from $R_1$ and the set $Z_{uv}$ to find in $R_2$ can be found by computing minimal red dominating set of $G(R,B)$.

    \needspace{2cm}
        
    \begin{claim}\label{claim:child2}
        For any set $D\in \D(G)$ such that $D^*=\parent_{uv}(D)$, let $X_{uv}$ and $Z_{uv}$ be the disjoint sets defined in the $\parent_{uv}$ relation, so that $D = (D^* \cup  \{u\} \cup Z_{uv}) \setminus (X_{uv} \cup \{v\})$. Then the set $R_1  \cup Z_{uv} \setminus (X_{uv} \cup \{v\})$ is a minimal red dominating set of $G(R,B)$.
    \end{claim}
    
    \begin{proof}
        Let $T=R_1  \cup Z_{uv} \setminus (X_{uv} \cup \{v\})$.
        By Claim~\ref{claim:child1}, we obtain that $T \subseteq (R_1 \setminus \{v\}) \cup R_2=R$. Therefore, it only remains to argue two things: that $T$ dominates $B$, and that $T$ is minimal, i.e., that every vertex in $T$ has a private neighbor in $B$ with respect to $T$.
        
        Let $y\in B$.
        By definition of $B$, we have $N(y)\cap D^*\subseteq R_1\cup \{v\}$ and $y\not\in N[u]$.
        Since $D$ dominates $y$, and given how $D$ and $D^*$ relate, the vertex $y$ has a neighbor either in $R_1\setminus (X_{uv}\cup \{v\})$ or in $Z_{uv}$. 
        In either case, $y$ has a neighbor in $T$. We conclude that $T$ dominates $B$.
        
        Let us now argue that every vertex $x$ in $T$ has a private neighbor in $B$ with respect to $T$. Note that $T\subseteq D$, hence that $x\in D$ and ${\priv(D,x)\neq\emptyset}$. 
        
        Let us first consider the case $x \in R_1 \setminus (X_{uv} \cup \{v\})$.  Since $u \in D$, we have $\priv(D,x)\subseteq N[x]\setminus N[u]$. Since moreover $x\in \upa{u}$ we have that $\priv(D,x)\subseteq \doa{x}\setminus N[u]$.
        In particular $x\not\in \priv(D,x)$.
        Let $y\in \priv(D,x)$.
        Then $N(y)\cap D=\{x\}$ and so $N(y)\cap D^*\subseteq \{x\}\cup X_{uv}\cup \{v\}\subseteq R_1$.
        We conclude by definition of $B$ that $y\in B$. 
        Therefore, every vertex in $T \cap (R_1 \setminus (X_{uv} \cup \{v\}))$ has a private neighbor in $B$ with respect to $T$. 
        
        We now consider the case $x\in Z_{uv}$. 
        Let $y \in \priv(D,x)$. 
        % Recall that $X_{uv}\cup \{v\}$ and $Z_{uv}$ are non-adjacent, and that $y\in N(X_{uv}\cup\{v\})$.
        % So $y \neq x$, and by Claim~\ref{claim:child1}, $y\in N(R_1)$.
        As $N(y)\cap D=\{x\}$ and $x\not\in D^*$, we have $N(y)\cap D^*\subseteq X_{uv}\cup\{v\}\subseteq R_1$.
        Hence $y\in B$.
        Consequently every $x\in T$ has a private neighbor in $B$, and so $T\in \D(R,B)$.
        \cqed
    \end{proof}    
    
    The core statement now follows easily.     
    
    \begin{claim}\label{claim:child}
        The set $\D$ contains every $D\in \D(G)$ such that $D^*=\parent_{uv}(D)$.
    \end{claim}
    
    \begin{proof}
        Let $D\in\D(G)$ be such that $D^*=\parent_{uv}(D)$.
        Then $D^*=((D \setminus \{u\}) \cup X_{uv} \cup \{v\})\setminus Z_{uv} $, where $X_{uv}$ and $Z_{uv}$ are the disjoint sets defined in the $\parent$ relation. 
        Consider the set $\smash{Y_{uv}=\bigcup_{z\in Z_{uv}}\priv(D,z)}$.
        
        From Claim~\ref{claim:child2}, we obtain that $R_1 \cup Z_{uv} \setminus (X_{uv} \cup \{v\})$ is a minimal dominating set of $G(R,B)$. Consequently, $\algob$ outputs $R_1 \cup Z_{uv} \setminus (X_{uv} \cup \{v\})$, which prompts $\algoa'$ to consider the set $(D^* \setminus R_1) \cup \{u\} \cup (R_1 \cup Z_{uv} \setminus (X_{uv} \cup \{v\}))=(D^* \cup \{u\} \cup Z_{uv})\setminus (X_{uv} \cup \{v\})=D$ as a candidate to be output after being greedily reduced into an irredundant set. Note that the set is already irredundant, so $D$ is generated. In other words, we have $D\in \D$ as desired.
        \cqed
    \end{proof}
    
    Each of the sets $R_1$, $R_2$, $R$ and $B$ can be constructed in polynomial time in $n$.
    The same goes for computing and reducing $D'$ into a minimal dominating set given $T\in \D(R,B)$.
    In addition, $R\cap \upa{\Max(B)}$ and $\Max(B)$ can be computed in polynomial time in $n$, and by Claim~\ref{claim:red-blue} the set $\D(R,B)$ can be generated with polynomial delay using $\algob$, concluding the proof.
\end{proof}

We conclude this section with the following corollary of Lemma~\ref{lemma:il-flip-golovach-improved} and Lemma~\ref{lemma:flipping-to-redblue}, observing that the antichain $B$ of a poset $P$ is minimal in $P[\upa{B}]$. 
Note that while the algorithm $\algoa'$ only computes the children $D$ for a fixed pair $u,v$, there are less than $n^2$ such pairs.
Checking for each candidate $D$ if it is an actual child (for a minimal pair $u,v$) is done in polynomial time as argued in Section~\ref{sec:flipping-operation}.
Moreover, to avoid repetitions, we only output a child $D$ if the smallest pair of vertices $a,b$ in $D$ such that $D^*=\parent_{ab}(D)$ coincides with $u,v$.
The obtained algorithm performs with polynomial delay (hence in incremental-polynomial time) and polynomial space, as desired.

\begin{theorem}\label{theorem:flipping-to-redblue}
    Let $\mathcal{G}$ be a graph class where every graph is comparability. If there is an incre\-mental-polynomial time and polynomial-space algorithm enumerating minimal red dominating sets in red-blue graphs of $\mathcal{G}$ whose blue vertices are minimal with respect to the associated poset, then there is one enumerating minimal dominating sets in graphs of $\mathcal{G}$.
\end{theorem}

\section{Red-blue domination in comparability graphs}\label{sec:red-blue}

We mentioned in Section~\ref{sec:preliminaries} that \textsc{Red-Blue-Dom-Enum} is already as hard as \textsc{Trans-Enum} even restricted to bipartite graphs, hence to comparability graphs.
We show nevertheless that the problem can be solved in in\-cremental-polynomial time under various restrictions on the red and blue sets (satisfying those of Lemma \ref{lemma:flipping-to-redblue}), as well as on the underlying poset.
More precisely, we show that, for any fixed integer $t$, \textsc{Red-Blue-Dom-Enum} is tractable in the comparability graph of $S_t$-free posets, whenever the blue elements are minimal in the poset. 
Since posets of bounded dimension do not contain any $S_p$ for some large enough $p$, we can derive the same for bounded dimension posets.

The key observation is that instances of red-blue domination in that case are of bounded conformality. As a corollary, we can use the algorithm of Khachiyan et al.~in~\cite{khachiyan2007dualization} to solve them in incremental-polynomial time. 
This yields by Theorem~\ref{theorem:flipping-to-redblue} an incremental-polynomial time algorithm enumerating minimal dominating sets in the comparability graphs of these posets.

Let us recall the notion of conformality introduced by Berge in~\cite{berge1984hypergraphs}. Informally, a hypergraph has small conformality when the property of not being contained in a hyperedge is witnessed by small subsets, in the sense that if a set is not contained in any hyperedge, then some small subset of it is not either.
More formally, let $c$ be an integer and $\H$ be a hypergraph.
We say that $\H$ is of \emph{conformality} $c$ if the following property holds for every subset $X\subseteq V(\H)$: $X$ is contained in a hyperedge of $\H$ whenever each subset of $X$ of cardinality at most $c$ is contained in a hyperedge of $\H$. 
Remember from Section~\ref{sec:preliminaries} that a hypergraph $\H$ is Sperner if $E_1\not\subseteq E_2$ for any two distinct hyperedges $E_1,E_2$ in $\H$.

Khachiyan, Boros, Elbassioni, and Gurvich proved the following; see also \cite{mary2024enumeration} for a simple and general proof of this theorem.

\begin{theorem}[\cite{khachiyan2007dualization,mary2024enumeration}]\label{theorem:conformality}
    The minimal transversals can be enumerated in incre\-mental-poly\-nomial time but using exponential space in Sperner hypergraphs of bounded conformality.
\end{theorem}

Our result is a corollary of the following, which basically says that in our setting, hypergraphs with large conformality induce large $S_t$ in the underlying poset.

\begin{lemma}\label{lemma:St}
    Let $P=(V,\leq)$ be a poset, $B=\Min(P)$ and $R=P-B$.
    Let $\H$ be the Sperner hypergraph defined by $V(\H)=P-B$ and 
    \[
        {\E(\H)=\Min_\subseteq\{\upa{x}\setminus \{x\} \mid x\in B\}}.
    \]
    If~$\H$ is not of conformality $t-1$ for some integer $t$, then $P$ contains $S_t$ as a suborder.
\end{lemma}

\begin{proof}
    Assume that $\H$ is not of conformality ${t-1}$, i.e., there is a subset $X\subseteq V(\H)$ that is not contained in a hyperedge of $\H$, and such that every subset $Y\subseteq X$ of size at most $t-1$ is contained in a hyperedge of $\H$.
    We consider $X=\{x_1,\dots, x_p\}$ of minimum cardinality.
    Then ${p\geq t}$ and to every $x_i\in X$ corresponds a hyperedge $E_i$ of $\H$ such that $E_i\cap X=\{X\setminus \{x_i\}\}$.
    Indeed, if no such $E_i$ exists for some $x_i\in X$, then $X'=X\setminus \{x_i\}$ is not contained in a hyperedge of $\H$, and still every subset $Y\subseteq X'$ of size at most ${t-1}$ is, contradicting the minimality of $X$.
    
    Let us show that $X$ is an antichain of $P$.
    Suppose toward a contradiction that $X$ is not an antichain and contains two elements $x_i,x_j$ such that ${x_i< x_j}$.
    As $R\subseteq \upa{B}\setminus B$, every hyperedge of $\H$ that contains $x_i$ contains $x_j$.
    We conclude that $X\subseteq E_j$, a contradiction. 
    Hence $X$ is an antichain.
    
    Consider now the antichain $\{e_1,\dots,e_p\}\subseteq B$ corresponding to $E_1,\dots,E_p$ in the poset $P$, i.e., such that ${E_i=\upa{e_i}\setminus \{e_i\}}$ for  $i\in\intv{1}{p}$.
    Then the set ${\{e_1,x_1,\dots,e_p,x_p\}}$ induces $S_p$ as suborder, $p\geq t$.
\end{proof}

Lemmas~\ref{lemma:il-flip-golovach-improved}, \ref{lemma:flipping-to-redblue}, and~\ref{lemma:St} together yield the following corollary.

\begin{corollary}\label{corollary:comp}
    There is an algorithm enumerating, for every fixed integer~$t$, the minimal dominating sets in comparability graphs of $S_t$-free posets.
\end{corollary}

Theorem~\ref{thm:maincomp} follows from Corollary~\ref{corollary:comp} and the observation that a poset containing $S_t$ has dimension at least $t$. 
Unfortunately, as the algorithms in~\cite{khachiyan2007dualization,mary2024enumeration} require exponential space, Corollary~\ref{corollary:comp} does not yield a poly\-nomial-space algorithm.

Finally, we note that while it is not clear whether the comparability graphs of bounded dimension posets are of bounded LMIM-width (and hence covered by the algorithm in~\cite{golovach2018lmimwidth} using similar methods), comparability graphs of $S_t$-free posets are not.

\section{Minimal dominating sets in incomparability graphs}\label{sec:flashlight}

We give a polynomial-delay algorithm enumerating minimal dominating sets in the incomparability graphs of bounded dimension posets, given with linear extensions witnessing the dimension. 
Although it is not necessary, we present a geometrical representation of the incomparability graphs that we find useful while working with them.

Let $P=(V,\leq)$ be a poset on $n$ elements and of dimension at most $d$. 
Let $\leq_1,\ldots,\leq_d$ be a sequence of linear orders witnessing it. 
Thus, we have $x \leq y$ in $P$ if and only if $x\leq_i y$ for each $i\in\{1,\ldots,d\}$.
Consider $d$ distinct vertical lines $\LL_1,\ldots, \LL_d$ in the plane, sorted from left to right in that order.
For each $i\in\{1,\ldots,d\}$, we distinguish $n$ points on $\LL_i$ and label them bottom-up 
with elements of $P$ sorted by $\leq_i$.
Now for each element $v$ in $P$, 
we define a piecewise linear curve $\overline{v}$ consisting of $d-1$ segments and connecting points labelled $v$ on consecutive lines.
It is a folklore observation, see e.g.\ \cite{golumbic1983comparability}, that the incomparability graph of $P$ is the intersection graph of this family of curves.
An example for $d=4$ is given in Figure~\ref{fig:lines}.

In the remaining of this section, we assume that a poset $P=(V,\leq)$ of dimension $d$ is given and we are also given the total orders $\leq_1,\ldots,\leq_d$ witnessing the dimension of $P$.
As described above we fix the lines $\LL_1,\ldots,\LL_d$, and the piecewise linear curves representing each element of $P$.
Let $G$ be the incomparability graph of $P$.

For a non-empty subset $S$ of elements of $P$ and $i\in\intv{1}{d}$, we define
$\LL_i(S)$ to be the maximum element of $S$ in $\leq_i$.
We call \emph{vertices upwards from $S$} the elements of the set 
\[
    \UU(S)=\{v\in V\setminus S \mid \LL_i(S) <_i v\ \text{for some}\ {i\in\intv{1}{d}}\}.
\]
A set $D$ is an \emph{upward extension} of $S$ if $S\subseteq D$ and $D\setminus S \subseteq \UU(S)$.

Let $S$ be a subset of elements of $P$ of size at least $3d$. 
We call the \emph{first layer} of $S$ the tuple $\AA(S)=(a_1,\dots,a_d)$ so that 
$a_1=\LL_1(S)$, and for every $i\in\intv{2}{d}$,
\[
    a_i=\LL_i\left(S\setminus \{a_1,\dots,a_{i-1}\}\right).
\]
Note that by this definition it might happen that $a_i\neq\LL_i(S)$.
This is in particular the case if $\LL_i(S)=\LL_j(S)$ for some $i,j\in\intv{1}{d}$ with $j<i$.
The \emph{second layer} of $S$ is the set $\BB(S)=\AA(S\setminus \AA(S))$.
The \emph{third layer} of $S$ is the set $\CC(S)=\AA\big(S\setminus (\AA(S)\cup \BB(S))\big)$.
Note that since $|S|\geq 3d$, the three layers are well-defined.
We call the \emph{border} of $S$ the concatenation 
\[
    \TT(S)=(a_1,\dots,a_d,b_1,\dots,b_d,c_1,\dots,c_d)
\]
of these three layers. 
In the following, we will consider irredundant such $S$. 
The reason for considering three layers is to help characterizing, in the coming proof, where the private neighbors of the elements in $\BB(S)$ lie, namely, ``between'' the elements of $\AA(S)$ and $\CC(S)$.

\begin{figure}
    \centering
    \includegraphics[scale=1.1]{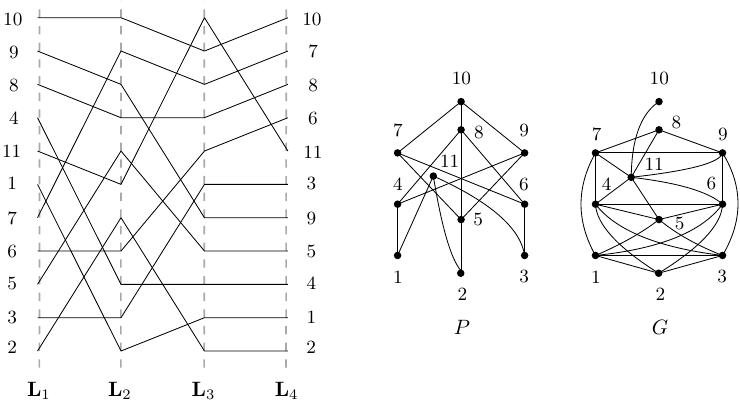}
    \caption{A poset $P$ and its incomparability graph $G$ as an intersection graph of curves induced by four linear extensions witnessing the dimension.}
    \label{fig:lines}
\end{figure}

We say that $I$ can be \emph{extended upwards into a minimal dominating set} whenever there is an upward extension of $I$ 
that is a minimal dominating set of $G$.
In the following, we aim to decide in polynomial time whether a given irredundant set $I$ in $G$  can be extended upwards into a minimal dominating set.
When the given set $I$ is of size at most $3d$, say $I=\{x_1,\dots,x_p\}$ and $p\leq 3d$, then this can be done efficiently
by checking for all the tuples $(y_1,\dots,y_p)\in \priv(I,x_1)\times \dots\times \priv(I,x_p)$ whether 
$\UU(I)\setminus N[y_1,\dots,y_p]$ dominates $G-N[I]$ or not.
A single tuple like that could be verified in $\Oh(n^2)$ time.
Since the number of tuples is no more than $n^{p}$, the total time is $\Oh(n^{3d+2})$.
If there is such a tuple $(y_1,\dots,y_p)$, then $\UU(I)\setminus N[y_1,\dots,y_p]$ can be greedily reduced into a minimal set $X$ so that $I\cup X$ is a minimal dominating set of $G$.
Otherwise, we know that $I$ cannot be extended as we explored all the possibilities for $I$ to keep its private neighbors in an upward extension.
This may be seen as an adaptation of the algorithm by Boros et al.~for the extension of subsets of constant size \cite{boros1998dual} to our ordered context.

The point of the following is to
show that the same technique can be applied for irredundant sets of arbitrary size.
The key insight is that it is enough to check whether we can extend $I$ into a dominating set so that all elements in the border $\TT(I)$ keep a private neighbor, with $\TT(I)$ being of constant size for fixed values of $d$.

\needspace{2cm}

\begin{theorem}\label{thm:bt-extension}
    Let $I$ be an irredundant set of $G$ of size at least $3d$ and 
    let $\TT(I)=(t_1,\ldots,t_{3d})$.
    Then $I$ can be extended upwards into a minimal dominating set of $G$ if and only if 
    there exists a tuple $(w_1,\dots,w_{3d})\in \priv(I,t_1)\times \dots\times \priv(I,t_{3d})$ such that
    $\UU(I)\setminus N[\{w_1,\dots,w_{3d}\}]$ is a dominating set of $G-N[I]$.
\end{theorem}

\begin{proof}
    First, we prove the forward implication. 
    Suppose that $I$ can be extended upwards into a minimal dominating set of $G$ and
    let $X\subseteq \UU(I)$ be such that $D=I\cup X$ is a minimal dominating set of $G$.
    Then $X$ dominates $G-N[I]$ and $\priv(D,u)\neq \emptyset$ for every $u\in D$. 
    Thus there exists a tuple $(w_1,\dots,w_{3d})$ in $\priv(D,t_1)\times \dots\times \priv(D,t_{3d})$, which, by definition, is such that $X\cap N[\{w_1,\dots,w_{3d}\}]=\emptyset$.
    Note that $\priv(D,t_i) \subseteq \priv(I,t_i)$ for each $i\in\{1,\ldots,3d\}$, hence that $(w_1,\dots,w_{3d})$ belongs to $\priv(I,t_1)\times \dots\times \priv(I,t_{3d})$ as well.
    This completes the proof of the forward implication.

    We now turn to the backward implication. 
    Intuitively, we aim to show in the following that a vertex from $\UU(I)$ may not steal a private neighbor of a vertex in $I\setminus \TT(I)$ without stealing the private neighbors of a vertex in $\TT(I)$, namely, one in $\BB(I)$.
    To prove this, we will show that the piecewise linear curves of the elements in $\TT(I)$ ``separate'' the curves of the private neighbors of the elements in $I\setminus \TT(I)$ from the curves of elements in $\UU(I)$.
    
    Suppose that there exists $(w_1,\dots,w_{3d})\in \priv(I,t_1)\times \dots\times \priv(I,t_{3d})$ such that
    $X:=\UU(I)\setminus N[\{w_1,\dots,w_{3d}\}]$ dominates $G-N[I]$.
    Thus $D=I\cup X$ dominates $G$. 
    In order to conclude that $I$ extends upwards into a minimal dominating set of $G$, all we need to see is that each element $u$ in $I$ has a private neighbor with respect to $D$, i.e., $\priv(D,u)\neq\emptyset$.
    Indeed, if so, then $D$ can be reduced into a minimal dominating set with the desired property that $D\setminus I\subseteq \UU(S)$.
    Since $X$ avoids $N[\{w_1,\dots,w_{3d}\}]$, we have that $w_i \in\priv(D,t_i)$ for each $i\in\{1,\ldots,3d\}$.
    Consider any element $u$ in $I\setminus \TT(I)$.
    Since $I$ is irredundant, we can fix $v\in\priv(I,u)$.
    We shall show that $v\in\priv(D,u)$.
    
    For convenience, we split the sequence $(t_1,\ldots,t_{3d})$ into three sequences $(a_1,\ldots,a_d)$, $(b_1,\ldots,b_d)$, and $(c_1,\ldots, c_d)$, so it relates to the initial layers $\AA(I)$, $\BB(I)$, and $\CC(I)$, respectively.
    
    In order to get a contradiction, suppose that there is $x\in X$ such that $x$ and $v$ are adjacent in $G$, i.e., $\overline{x}$ and $\overline{v}$ intersect.
    In particular, there is some $q\in \{1,\ldots,d\}$ such that $x <_q v$. 

    We claim that
    \[
        v < b_k \text{ in $P$ for every $k\in\{1,\ldots,d\}$.}
    \]
    Since $v\in \priv(I,u)$ and $b_k \in I$, $\overline{v}$ and $\overline{b_k}$ do not intersect. % for any $b\in \BB(I)$.
    Therefore, either $v < b_k$ in $P$ or $v > b_k$ in $P$.
    Assume toward a contradiction that $b_k < v$ in $P$.
    Since $\overline{u}$ and $\overline{v}$ intersect, we can fix $p\in\{1,\ldots,d\}$ such that $v<_p u$. 
    Recall that $u \in I \setminus \TT(I)$. 
    Thus, by the definition of the third layer $\CC(I)$ we have $u <_p c_p$.
    Hence $v <_p u <_p c_p$. 
    Since $v\in\priv(I,u)$ and $c\in I$, we have that $\overline{c_p}$ and $\overline{v}$ are disjoint. 
    We deduce that $v < c_p$ in $P$.
    This contradicts our assumption as $b_k < v < c_p$ in $P$ but, by construction, no element of the second layer can be below an element of the third layer. This completes the proof of the claim.

    In particular, we have $x <_q v <_q b_k$ for every $k\in\{1,\ldots,d\}$.

    Consider the tuple $(s_1,\dots,s_d)=(w_{d+1},\ldots,w_{2d})$ of private neighbors of the elements of $(b_1,\ldots,b_{d})$.
    We claim that there exist indices $\alpha,\beta\in\{1,\ldots,d\}$ such that
    \[
        b_{\alpha} <_q s_{\beta}.
    \]
    Towards the contradiction, assume that $s_i <_q b_j$ for all $i,j\in\{1,\ldots,d\}$.
    Since $\overline{s_i}$ and $\overline{b_j}$ are non-intersecting for $i\neq j$, we conclude that $s_i < b_j$ in $P$ for every $i\neq j$.
    Since $\overline{s_i}$ intersects $\overline{b_i}$, there is an index $t(i)$ so that $b_i <_{t(i)} s_i$, for each $i\in\{1,\ldots,d\}$. 
    Clearly, the values $t(i)$ must be all distinct for $i\in\{1,\ldots,d\}$. 
    This way, we need to take $d$ distinct values for $t_i$'s and because of our assumption all of them are in $\{1,\ldots,d\}\setminus\{q\}$. We obtain a contradiction, proving the claim.

    We concluded so far that
    \[
        x <_q v <_q b_{\alpha} <_q s_{\beta}.
    \]

    Now recall that $x \in X \subseteq\UU(I)$. 
    Thus, there must be some $p\in\{1,\ldots,d\}$ with
    \[
        \LL_p(I) <_p x.
    \]
    Recall also that $\overline{s_{\beta}}$ intersects $\overline{b_{\beta}}$, 
    so we can fix $t\in\{1,\ldots,d\}$ such that $s_{\beta} <_t b_{\beta}$. 
    By the definition of the first two layers, we have that $b_{\beta} <_t a_t$. 
    Thus, $s_{\beta} <_t a_t$. 
    Since $s_{\beta}\in\priv(I,b_{\beta})$ and $a_t\in I$, we get that $\overline{s_{\beta}}$ and $\overline{a_t}$ are non-intersecting. 
    Therefore, $s_{\beta} < a_t$ in $P$. 
    In particular, we get $s_{\beta} <_p a_t$ and
    \[
        s_{\beta} <_p a_t \leq_p \LL_p(I) <_p x.
    \]
    
    The two inequalities $s_{\beta} >_q x$ and $s_{\beta} <_p x$ imply that $\overline{x}$ intersects $\overline{s_{\beta}}$. 
    Thus $x$ and $s_{\beta}$ are adjacent in $G$. 
    This contradicts the assumption that $X \subseteq \UU(I) \setminus N[s_{\beta}]$ and completes the proof of the backward implication.
\end{proof}

We deduce the next corollary, by guessing a good tuple $(w_1,\dots,w_{3d})$ as in Theorem~\ref{thm:bt-extension} in case when $|I|\geq 3d$, and checking for such a tuple whether the set $\UU(I)\setminus N[\{w_1,\dots,w_{3d}\}]$ is a dominating set of $G-N[I]$.
Space is polynomial as we only iterate through neighborhoods. 

\begin{corollary}\label{cor:bt-extension}
    There is an algorithm that, given an irredundant set $I$ of $G$, decides in $\Oh(n^{3d+2})$ time and polynomial space whether $I$ can be extended upwards into a minimal dominating set.
\end{corollary}

We are now ready to describe a simple algorithm enumerating minimal dominating sets in the incomparability graphs of bounded dimension posets.
The algorithm will construct solutions one vertex at a time, guaranteeing that every partial solution extends to a minimal dominating set, in a fashion similar to the \emph{flashlight search} technique~\cite{strozecki2019efficient,capelli2021enumerating}.
The following $\parent$ relation between partial solutions will be used by the algorithm to handle repetitions.

\begin{definition}\label{def:parent}
    Let $I$ be a non-empty irredundant set of $G$ that can be extended upwards into a minimal dominating set.
    We call parent of $I$ the unique irredundant set $I^*=\parent(I)$ obtained by removing vertex $\LL_1(I)$ from $I$, i.e., the greatest vertex $v$ in $I$ with respect to $<_1$.
\end{definition}

Observe that every minimal dominating set $D$ of $G$ is an irredundant set of $G$ that can be extended upwards into a minimal dominating set (the extension being $D$ itself), with no children.
Conversely, an irredundant set that can be extended upwards into a minimal dominating set, with no children, is necessarily a minimal dominating set.

Consequently, the $\parent$ relation as introduced in Definition~\ref{def:parent} defines a tree $T$ whose nodes are irredundant sets of $G$ that can be extended upwards into minimal dominating sets, root is the empty set, leaves are minimal dominating sets of $G$, and where there is an edge between two irredundant sets $I^*$ and $I$ if $I^*=\parent(I)$. 
The fact that all minimal dominating sets appear as leaves of this tree follows by induction noting that each node of $T$ distinct from the root has a parent whose cardinal is strictly smaller.
As in Section~\ref{sec:polyflipping}, the enumeration proceeds with what boils down to a DFS of $T$ initiated at the empty set.
When visiting a node $I^*\in V(T)$, the algorithm seeks the children of $I$ as follows.
It checks for every candidate vertex $v\in \UU(I^*)$ whether $I^*\cup \{v\}$ can be extended upwards into a minimal dominating set, using Corollary~\ref{cor:bt-extension}, and whether the obtained set is a child of $I^*$, using Definition~\ref{def:parent}.
Whenever it is the case, the algorithm ``pauses'' the generation of children of $I^*$, and generates children of $I^*\cup\{v\}$. 
When the generation on $I^*\cup\{v\}$ is complete, the algorithm ``resumes'' the generation on $I^*$.
During this procedure, only the leaves of $T$, hence the minimal dominating sets of $G$, are output by the algorithm.
Duplications are implicitly avoided by the structure of $T$.

The delay time complexity is bounded by twice the depth of the tree (the maximal distance between two leaves in $T$), times the time complexity of solving the extension problem and checking the $\parent$ relation, for every candidate vertex $v$.
This sums up to
\[
    2n \cdot \Oh(n^{3d+2}+n^2) \cdot n = \Oh(n^{3d+4})
\]
Space complexity is polynomial as we only need to store for each node $W\in T$ from the root to the current node $I^*\in T$ the data of the (paused) execution of the children generation on node $W$. 

We conclude to Theorem~\ref{thm:mainincomp} that we restate here with the complexity.

\begin{theorem}\label{theorem:incomp}
    There is an algorithm that, given the incomparability graph $G$ of an $n$-element poset $P$ of dimension $d$, together with $d$ linear orders witnessing the dimension of $P$, enumerates all minimal dominating sets of $G$ with delay $\Oh(n^{3d+4})$ and using polynomial space. 
\end{theorem}

% It should however be noted that incomparability graphs of posets of dimension $d$ are $d$-trapezoid graphs, which were shown in~\cite{golovach2018lmimwidth} to admit a linear-delay algorithm after $n^{O(1)}$-time preprocessing if given with their parallel lines representation.
% We nevertheless point that our algorithm is remarkably simpler that the one obtained in \cite{golovach2018lmimwidth} and enjoys explicit time-complexity bounds, which may be of practical interest for implementations purposes.

\section{Discussions}\label{sec:conclusion}

In this paper, we provided an incremental-polynomial (resp.~polynomial-delay) algorithm enumerating the minimal dominating sets in the comparability (resp.~incomparability) graphs of bounded dimension posets.
Recall that the complexity of \textsc{Trans-Enum}, which reduces to \textsc{Dom-Enum} in co-bipartite graphs, is widely open. 
As incomparability graphs include co-bipartite graphs, dropping the dimension in Theorem~\ref{thm:mainincomp} is one of the most important algorithmic challenges in enumeration.
On the other hand, dropping the dimension in Theorem~\ref{thm:maincomp} seems a more tractable, though fascinating challenge.

The algorithm provided by Corollary~\ref{corollary:comp} covers in fact all comparability graphs of $S_t$-free posets, not just those of bounded dimension. The two first authors initiated in \cite{bonamy2020kt} a characterization of the complexity of \textsc{Dom-Enum} in $H$-free graphs. Despite the fact that they encapsulate all the hardness of \textsc{Dom-Enum}, incomparability graphs have a restrictive structure. 
What can we say about $H$-free incomparability graphs? 
Due to co-bipartite graphs, this question is only interesting for a co-bipartite $H$. 
To match Theorem~\ref{thm:maincomp}, it is tempting to wonder what happens for incomparability graphs of $S_t$-free posets.
Our proof clearly does not extend naturally, as it crucially relies on the bounded dimension representation---recall that not all $S_t$-free posets have small dimension \cite{Tro-book}.
However, we believe that an extension would be possible with different tools.

\begin{conjecture}\label{conj:Stfreeincomps}
    For every $t$, there is an output-polynomial time algorithm for \textsc{Dom-Enum} in incomparability graphs of $S_t$-free posets. 
\end{conjecture}

A directly easier conjecture is the following:

\begin{conjecture}\label{conj:2dpfreeincomps}
    For every $p$, there is an output-polynomial time algorithm for \textsc{Dom-Enum} in $2K_p$-free incomparability graphs.
\end{conjecture}

Indeed, note that an $S_{2p}$ suborder in a poset yields a $2K_p$ in the corresponding incomparability graph, where $2K_p$ denotes two disjoint cliques on $p$ elements.
Symmetrically, Corollary \ref{corollary:comp} implies that for every $p$, there is an output-polynomial time algorithm for \textsc{Dom-Enum} in $K_{p,p}$-free comparability graphs. 
Conjecture~\ref{conj:2dpfreeincomps} holds trivially for $p=1$ (all such graphs are cliques). 
It is in fact also true for $p=2$, as $2K_2$-free graphs admit an output-polynomial time algorithm~\cite{bonamy2020kt}. 
The case $p=3$ is widely open without the incomparability assumption~\cite{bonamy2020kt}.

Questions abound around Conjecture~\ref{conj:2dpfreeincomps}, mainly by considering other simple cases of a co-bipartite $H$. 
The case $H=C_4$ is the smallest open case for general graphs, but $C_4$-free incomparability graphs are interval graphs~\cite{gilmore1964characterization}; there \textsc{Dom-Enum} admits a linear delay algorithm \cite{kante2013permutation}. A natural generalization would be the case of $H$ being two cliques with a matching between them, or $H$ being a clique minus a matching. Instead of debating which generalizations are most sensible, we state the following bold conjecture:

\begin{conjecture}\label{conj:cobipfreeincomps}
    For any co-bipartite $H$, there is an output-polynomial time algorithm for \textsc{Dom-Enum} in $H$-free incomparability graphs.
\end{conjecture}

There is in fact no blatant reason why the incomparability condition would be necessary here. Other natural parameters for the classes considered in this paper concern the height and width of the poset.
The \emph{height} (resp.~\emph{width}) of a poset is the size of its maximum chain (resp.~antichain).
Since those are related to the largest size of a clique, the algorithm in~\cite{bonamy2020kt} covers the comparability graphs of posets of bounded height and the incomparability graphs of posets of bounded width.
We can show easily that \textsc{Dom-Enum} is tractable in the comparability graphs of posets of bounded width.

\begin{proposition}\label{prop:alpha-bound}
    Let $G$ be the comparability graph of a poset $P=(V,\leq)$ of width $\alpha$, and $D$ be a minimal dominating set of $G$.
    Then $|D|\leq 2\alpha$.
\end{proposition}

\begin{proof}
    Consider for a contradiction a minimal dominating set $D$ of $G$ with $|D|> 2 \alpha$. Then $D$ contains three elements $x,y,z$ such that $x<y<z$. As a consequence, $N[y]\subseteq N[x]\cup N[z]$, which contradicts $\priv(D,y) \neq \emptyset$.
\end{proof}

Proposition~\ref{prop:alpha-bound} guarantees that a brute-force test of all small subsets yields a polynomial-time algorithm enumerating minimal dominating sets in the comparability graph of posets of bounded width.

As for incomparability graphs of posets of height $2$, we consider the incompatibilities of a bipartite order and observe that all co-bipartite graphs can be represented as incomparability graphs of a poset of height at most $2$. 
Therefore, bounded height is not a helpful parameter for incomparability graphs. A more interesting question is perhaps whether restricting posets to lattices yields efficient algorithms, both for comparability and incomparability graphs.

Our algorithm for incomparability graphs of bounded dimension relies heavily on their geometric representation, as in~\cite{golovach2018lmimwidth}. An open question is whether we can achieve polynomial delay without the representation.

\bibliographystyle{alpha}
\bibliography{main}

@article{golovach2016chordalbip,
  title={Enumerating minimal dominating sets in chordal bipartite graphs},
  author={Golovach, Petr A. and Heggernes, Pinar and Kant{\'e}, Mamadou M. and Kratsch, Dieter and Villanger, Yngve},
  journal={Discrete Applied Mathematics},
  volume={199},
  pages={30--36},
  year={2016},
  publisher={Elsevier}
}

@article{golovach2015flipping,
  title={An incremental polynomial time algorithm to enumerate all minimal edge dominating sets},
  author={Golovach, Petr A. and Heggernes, Pinar and Kratsch, Dieter and Villanger, Yngve},
  journal={Algorithmica},
  volume={72},
  number={3},
  pages={836--859},
  year={2015},
  publisher={Springer}
}

@article{khachiyan2007dualization,
  title={On the dualization of hypergraphs with bounded edge-intersections and other related classes of hypergraphs},
  author={Khachiyan, Leonid and Boros, Endre and Elbassioni, Khaled and Gurvich, Vladimir},
  journal={Theoretical Computer Science},
  volume={382},
  number={2},
  pages={139--150},
  year={2007},
  publisher={Elsevier}
}

@book{berge1984hypergraphs,
  title={Hypergraphs: combinatorics of finite sets},
  author={Berge, Claude},
  volume={45},
  year={1984},
  publisher={Elsevier}
}

@article{johnson1988generating,
  title={On generating all maximal independent sets},
  author={Johnson, David S. and Yannakakis, Mihalis and Papadimitriou, Christos H.},
  journal={Information Processing Letters},
  volume={27},
  number={3},
  pages={119--123},
  year={1988},
  publisher={Elsevier}
}

@article{bonamy2020kt,
  title={Enumerating minimal dominating sets in {$K_t$}-free graphs and variants},
  author={Bonamy, Marthe and Defrain, Oscar and Heinrich, Marc and Pilipczuk, Micha{\l} and Raymond, Jean-Florent},
  journal={ACM Transactions on Algorithms (TALG)},
  volume={16},
  number={3},
  pages={1--23},
  year={2020},
  publisher={ACM New York, NY, USA}
}

@article{kante2014split,
    author = "Mamadou M. {Kant\'e} and Vincent Limouzy and Arnaud Mary and Lhouari Nourine",
    title = "On the Enumeration of Minimal Dominating Sets and Related Notions",
    journal = "SIAM Journal on Discrete Mathematics",
    year = "2014",
    pages = "1916-1929",
    volume = "28(4)",
}

@article{fredman1996complexity,
    author = "Michael L. Fredman and Leonid Khachiyan",
    title = "On the Complexity of Dualization of Monotone Disjunctive Normal Forms",
    journal = "Journal of Algorithms",
    volume = "21",
    number = "3",
    pages = "618-628",
    year = "1996",
    issn = "0196-6774",
}

@article{golovach2018lmimwidth,
  title={Output-polynomial enumeration on graphs of bounded (local) linear {MIM}-width},
  author={Golovach, Petr A. and Heggernes, Pinar and Kant{\'e}, Mamadou M. and Kratsch, Dieter and S{\ae}ther, Sigve H. and Villanger, Yngve},
  journal={Algorithmica},
  volume={80},
  number={2},
  pages={714--741},
  year={2018},
  publisher={Springer}
}

@article{golumbic1983comparability,
  title={Comparability graphs and intersection graphs},
  author={Golumbic, Martin C. and Rotem, Doron and Urrutia, Jorge},
  journal={Discrete Mathematics},
  volume={43},
  number={1},
  pages={37--46},
  year={1983},
  publisher={Elsevier}
}

@article{strozecki2019efficient,
  title={Efficient enumeration of solutions produced by closure operations},
  author={Strozecki, Yann and Mary, Arnaud},
  journal={Discrete Mathematics \& Theoretical Computer Science},
  volume={21},
  year={2019},
  publisher={Episciences. org}
}

@article{strozecki2019survey,
  title={Enumeration Complexity},
  author={Strozecki, Yann},
  journal={Bulletin of EATCS},
  volume={1},
  number={129},
  year={2019}
}

@article{gilmore1964characterization,
  title={A characterization of comparability graphs and of interval graphs},
  author={Gilmore, Paul C. and Hoffman, Alan J.},
  journal={Canadian Journal of Mathematics},
  volume={16},
  pages={539--548},
  year={1964},
  publisher={Cambridge University Press}
}

@article{creignou2019complexity,
  title={A complexity theory for hard enumeration problems},
  author={Creignou, Nadia and Kr{\"o}ll, Markus and Pichler, Reinhard and Skritek, Sebastian and Vollmer, Heribert},
  journal={Discrete Applied Mathematics},
  year={2019},
  publisher={Elsevier}
}

@article{eiter2003new,
  title={New results on monotone dualization and generating hypergraph transversals},
  author={Eiter, Thomas and Gottlob, Georg and Makino, Kazuhisa},
  journal={SIAM Journal on Computing},
  volume={32},
  number={2},
  pages={514--537},
  year={2003},
  publisher={SIAM}
}

@inproceedings{bonamy2019triangle,
  title={Enumerating minimal dominating sets in triangle-free graphs},
  author={Bonamy, Marthe and Defrain, Oscar and Heinrich, Marc and Raymond, Jean-Florent},
  booktitle={36th International Symposium on Theoretical Aspects of Computer Science},
  year={2019},
  organization={Schloss Dagstuhl-Leibniz-Zentrum fuer Informatik}
}

@inproceedings{boros2004generating,
  title={Generating maximal independent sets for hypergraphs with bounded edge-intersections},
  author={Boros, Endre and Elbassioni, Khaled and Gurvich, Vladimir and Khachiyan, Leonid},
  booktitle={Latin American Symposium on Theoretical Informatics},
  pages={488--498},
  year={2004},
  organization={Springer}
}

@inproceedings{kante2018holes,
  title={Enumerating Minimal Transversals of Hypergraphs without Small Holes},
  author={Kant{\'e}, Mamadou M. and Khoshkhah, Kaveh and Pourmoradnasseri, Mozhgan},
  booktitle={43rd International Symposium on Mathematical Foundations of Computer Science},
  year={2018},
  organization={Schloss Dagstuhl-Leibniz-Zentrum fuer Informatik}
}

@inproceedings{kante2015chordal,
  title={A polynomial delay algorithm for enumerating minimal dominating sets in chordal graphs},
  author={Kant{\'e}, Mamadou M. and Limouzy, Vincent and Mary, Arnaud and Nourine, Lhouari and Uno, Takeaki},
  booktitle={International Workshop on Graph-Theoretic Concepts in Computer Science},
  pages={138--153},
  year={2015},
  organization={Springer}
}

@inproceedings{kante2015line,
  title={Polynomial delay algorithm for listing minimal edge dominating sets in graphs},
  author={Kant{\'e}, Mamadou M. and Limouzy, Vincent and Mary, Arnaud and Nourine, Lhouari and Uno, Takeaki},
  booktitle={Workshop on Algorithms and Data Structures},
  pages={446--457},
  year={2015},
  organization={Springer}
}

@inproceedings{defrain2019neighborhood,
  title={Neighborhood Inclusions for Minimal Dominating Sets Enumeration: Linear and Polynomial Delay Algorithms in ${P}_7$-Free and ${P}_8$-Free Chordal Graphs},
  author={Defrain, Oscar and Nourine, Lhouari},
  booktitle={30th International Symposium on Algorithms and Computation},
  year={2019},
  organization={Schloss Dagstuhl-Leibniz-Zentrum fuer Informatik}
}

@inproceedings{kante2013permutation,
  title={On the enumeration and counting of minimal dominating sets in interval and permutation graphs},
  author={Kant{\'e}, Mamadou M. and Limouzy, Vincent and Mary, Arnaud and Nourine, Lhouari and Uno, Takeaki},
  booktitle={International Symposium on Algorithms and Computation},
  pages={339--349},
  year={2013},
  organization={Springer}
}

@article{courcelle2009linear,
  title={Linear delay enumeration and monadic second-order logic},
  author={Courcelle, Bruno},
  journal={Discrete Applied Mathematics},
  volume={157},
  number={12},
  pages={2675--2700},
  year={2009},
  publisher={Elsevier}
}

@article{dushnik1941partially,
  title={Partially ordered sets},
  author={Dushnik, Ben and Miller, Edwin W.},
  journal={American journal of mathematics},
  volume={63},
  number={3},
  pages={600--610},
  year={1941},
  publisher={JSTOR}
}

@article{eiter1995identifying,
  title={Identifying the minimal transversals of a hypergraph and related problems},
  author={Eiter, Thomas and Gottlob, Georg},
  journal={SIAM Journal on Computing},
  volume={24},
  number={6},
  pages={1278--1304},
  year={1995},
  publisher={SIAM}
}

@article{eiter2008computational,
  title={Computational aspects of monotone dualization: A brief survey},
  author={Eiter, Thomas and Makino, Kazuhisa and Gottlob, Georg},
  journal={Discrete Applied Mathematics},
  volume={156},
  number={11},
  pages={2035--2049},
  year={2008},
  publisher={Elsevier}
}

@article{fomin2008combinatorial,
  title={Combinatorial bounds via measure and conquer: Bounding minimal dominating sets and applications},
  author={Fomin, Fedor V. and Grandoni, Fabrizio and Pyatkin, Artem V. and Stepanov, Alexey A.},
  journal={ACM Transactions on Algorithms},
  volume={5},
  number={1},
  pages={9},
  year={2008},
  publisher={ACM}
}

@article{couturier2013minimal,
  title={Minimal dominating sets in graph classes: combinatorial bounds and enumeration},
  author={Couturier, Jean-Fran{\c{c}}ois and Heggernes, Pinar and van’t Hof, Pim and Kratsch, Dieter},
  journal={Theoretical Computer Science},
  volume={487},
  pages={82--94},
  year={2013},
  publisher={Elsevier}
}

@article{golovach2019input,
  title={Enumeration and maximum number of minimal dominating sets for chordal graphs},
  author={Golovach, Petr A. and Kratsch, Dieter and Liedloff, Mathieu and Sayadi, Mohamed Yosri},
  journal={Theoretical Computer Science},
  volume={783},
  pages={41--52},
  year={2019},
  publisher={Elsevier}
}

@article{capelli2021enumerating,
  title={Enumerating models of DNF faster: breaking the dependency on the formula size},
  author={Capelli, Florent and Strozecki, Yann},
  journal={Discrete Applied Mathematics},
  volume={303},
  pages={203--215},
  year={2021},
  publisher={Elsevier}
}

@article{tsukiyama1977new,
  title={A new algorithm for generating all the maximal independent sets},
  author={Tsukiyama, Shuji and Ide, Mikio and Ariyoshi, Hiromu and Shirakawa, Isao},
  journal={SIAM Journal on Computing},
  volume={6},
  number={3},
  pages={505--517},
  year={1977},
  publisher={SIAM}
}

@inproceedings{makino2004new,
  title={New algorithms for enumerating all maximal cliques},
  author={Makino, Kazuhisa and Uno, Takeaki},
  booktitle={Scandinavian workshop on algorithm theory},
  pages={260--272},
  year={2004},
  organization={Springer}
}

@book{Tro-book,
	Address = {Baltimore, MD},
	Author = {Trotter, William T.},
	Isbn = {0-8018-4425-8},
	Mrclass = {06-02 (05-02 06A06 06A07)},
	Mrnumber = {1169299 (94a:06001)},
	Mrreviewer = {K. P. Bogart},
	Note = {Dimension theory},
	Pages = {xvi+307},
	Publisher = {Johns Hopkins University Press},
	Series = {Johns Hopkins Series in the Mathematical Sciences},
	Title = {Combinatorics and partially ordered sets},
	Year = {1992}}

@inproceedings{capelli2023geometric,
  title={Geometric Amortization of Enumeration Algorithms},
  author={Capelli, Florent and Strozecki, Yann},
  booktitle={40th International Symposium on Theoretical Aspects of Computer Science (STACS 2023)},
  year={2023},
  organization={Schloss Dagstuhl--Leibniz-Zentrum f{\"u}r Informatik}
}

@article{mary2024enumeration,
  title={Enumeration of minimal transversals of hypergraphs of bounded VC-dimension},
  author={Mary, Arnaud},
  journal={arXiv preprint arXiv:2407.00694},
  year={2024}
}

@article{boros1998dual,
  title={Dual subimplicants of positive Boolean functions},
  author={Boros, Endre and Gurvich, Vladimir and Hammer, Peter L},
  journal={Optimization Methods and Software},
  volume={10},
  number={2},
  pages={147--156},
  year={1998},
  publisher={Taylor \& Francis}
}

\end{document}